\RequirePackage[l2tabu,orthodox]{nag}
\documentclass
[letterpaper,11pt,]
{article}

\usepackage{etex}
\usepackage{verbatim}
\usepackage{xspace,enumerate}
\usepackage[dvipsnames]{xcolor}
\usepackage[T1]{fontenc}
\usepackage[full]{textcomp}
\usepackage[american]{babel}
\usepackage{mathtools}
\usepackage{amsthm}
\usepackage[
letterpaper,
top=1in,
bottom=1in,
left=1in,
right=1in]{geometry}
\usepackage{newpxtext} 
\usepackage{textcomp} 
\usepackage[varg,bigdelims]{newpxmath}
\usepackage[scr=rsfso]{mathalfa}
\usepackage{bm} 
\let\mathbb\varmathbb
\usepackage{microtype}
\usepackage[pagebackref,colorlinks=true,urlcolor=blue,linkcolor=blue,citecolor=OliveGreen]{hyperref}
\usepackage[capitalise,nameinlink]{cleveref}
\crefname{lemma}{Lemma}{Lemmas}
\crefname{fact}{Fact}{Facts}
\crefname{theorem}{Theorem}{Theorems}
\crefname{corollary}{Corollary}{Corollaries}
\crefname{claim}{Claim}{Claims}
\crefname{example}{Example}{Examples}
\crefname{algorithm}{Algorithm}{Algorithms}
\crefname{problem}{Problem}{Problems}
\crefname{definition}{Definition}{Definitions}
\crefname{exercise}{Exercise}{Exercises}
\usepackage{amsthm}

\newtheorem{theorem}{Theorem}[section]
\newtheorem*{theorem*}{Theorem}
\newtheorem{lemma}[theorem]{Lemma}
\newtheorem*{lemma*}{Lemma}
\newtheorem{fact}[theorem]{Fact}
\newtheorem*{fact*}{Fact}

\newtheorem*{proposition*}{Proposition}
\newtheorem{corollary}[theorem]{Corollary}
\newtheorem*{corollary*}{Corollary}

\newtheorem*{hypothesis*}{Hypothesis}

\newtheorem*{conjecture*}{Conjecture}
\theoremstyle{definition}
\newtheorem{definition}[theorem]{Definition}
\newtheorem*{definition*}{Definition}

\newtheorem*{construction*}{Construction}

\newtheorem*{example*}{Example}

\newtheorem*{question*}{Question}

\newtheorem*{algorithm*}{Algorithm}

\newtheorem*{assumption*}{Assumption}

\newtheorem*{problem*}{Problem}

\newtheorem*{openquestion*}{Open Question}
\theoremstyle{remark}

\newtheorem*{claim*}{Claim}

\newtheorem*{remark*}{Remark}
\newtheorem{observation}[theorem]{Observation}
\newtheorem*{observation*}{Observation}
\usepackage{paralist}
\let\originalleft\left
\let\originalright\right
\renewcommand{\left}{\mathopen{}\mathclose\bgroup\originalleft}
\renewcommand{\right}{\aftergroup\egroup\originalright}
\usepackage{turnstile}
\usepackage{mdframed}
\usepackage{tikz}
\usepackage{caption}
\DeclareCaptionType{Algorithm}
\usepackage{newfloat}
\usepackage{xparse}
\usepackage{amsthm} 
\makeatletter
\let\latexparagraph\paragraph
\RenewDocumentCommand{\paragraph}{som}{%
  \IfBooleanTF{#1}
    {\latexparagraph*{#3}}
    {\IfNoValueTF{#2}
       {\latexparagraph{\maybe@addperiod{#3}}}
       {\latexparagraph[#2]{\maybe@addperiod{#3}}}%
  }%
}
\newcommand{\maybe@addperiod}[1]{%
  #1\@addpunct{.}%
}
\makeatother

\newcommand{\Authornote}[2]{}
\newcommand{\Authornotecolored}[3]{}
\newcommand{\Authorcomment}[2]{}
\newcommand{\Authorfnote}[2]{}

\usepackage{boxedminipage}
\newcommand{\paren}[1]{(#1)}
\newcommand{\Paren}[1]{\left(#1\right)}

\newcommand{\brac}[1]{[#1]}
\newcommand{\Brac}[1]{\left[#1\right]}

\newcommand{\abs}[1]{\lvert#1\rvert}

\newcommand{\card}[1]{\lvert#1\rvert}

\newcommand{\set}[1]{\{#1\}}
\newcommand{\Set}[1]{\left\{#1\right\}}

\newcommand{\norm}[1]{\lVert#1\rVert}
\newcommand{\Norm}[1]{\left\lVert#1\right\rVert}

\newcommand{\normt}[1]{\norm{#1}_2}
\newcommand{\Normt}[1]{\Norm{#1}_2}

\newcommand{\snormt}[1]{\norm{#1}^2_2}
\newcommand{\Snormt}[1]{\Norm{#1}^2_2}

\newcommand{\Snorm}[1]{\Norm{#1}^2}

\newcommand{\iprod}[1]{\langle#1\rangle}
\newcommand{\Iprod}[1]{\left\langle#1\right\rangle}

\newcommand{\Esymb}{\mathbb{E}}
\newcommand{\Psymb}{\mathbb{P}}

\newcommand{\suchthat}{\;\middle\vert\;}

\newcommand{\sle}{\preceq}

\newcommand\bdot\bullet

\newcommand{\N}{\mathbb N}
\newcommand{\R}{\mathbb R}

\newcommand{\cA}{\mathcal A}
\newcommand{\cB}{\mathcal B}

\newcommand{\cD}{\mathcal D}

\newcommand{\cN}{\mathcal N}
\newcommand{\cO}{\mathcal O}

\renewcommand{\leq}{\leqslant}

\renewcommand{\geq}{\geqslant}

\let\epsilon=\varepsilon
\numberwithin{equation}{section}
\newcommand\MYcurrentlabel{xxx}
\newcommand{\MYstore}[2]{%
  \global\expandafter \def \csname MYMEMORY #1 \endcsname{#2}%
}
\newcommand{\MYload}[1]{%
  \csname MYMEMORY #1 \endcsname%
}
\newcommand{\MYnewlabel}[1]{%
  \renewcommand\MYcurrentlabel{#1}%
  \MYoldlabel{#1}%
}
\newcommand{\MYdummylabel}[1]{}
\newcommand{\torestate}[1]{%
  \let\MYoldlabel\label%
  \let\label\MYnewlabel%
  #1%
  \MYstore{\MYcurrentlabel}{#1}%
  \let\label\MYoldlabel%
}
\newcommand{\restatetheorem}[1]{%
  \let\MYoldlabel\label
  \let\label\MYdummylabel
  \begin{theorem*}[Restatement of \cref{#1}]
    \MYload{#1}
  \end{theorem*}
  \let\label\MYoldlabel
}
\newcommand{\restatelemma}[1]{%
  \let\MYoldlabel\label
  \let\label\MYdummylabel
  \begin{lemma*}[Restatement of \cref{#1}]
    \MYload{#1}
  \end{lemma*}
  \let\label\MYoldlabel
}
\newcommand{\restateprop}[1]{%
  \let\MYoldlabel\label
  \let\label\MYdummylabel
  \begin{proposition*}[Restatement of \cref{#1}]
    \MYload{#1}
  \end{proposition*}
  \let\label\MYoldlabel
}
\newcommand{\restatefact}[1]{%
  \let\MYoldlabel\label
  \let\label\MYdummylabel
  \begin{fact*}[Restatement of \cref{#1}]
    \MYload{#1}
  \end{fact*}
  \let\label\MYoldlabel
}
\newcommand{\restate}[1]{%
  \let\MYoldlabel\label
  \let\label\MYdummylabel
  \MYload{#1}
  \let\label\MYoldlabel
}

\newcommand{\sse}{\subseteq}

\newcommand{\e}{\epsilon}
\newcommand{\eps}{\epsilon}

\allowdisplaybreaks
\DeclarePairedDelimiterX{\infdivx}[2]{(}{)}{#1\;\delimsize\|\;#2}
\newcommand{\kldiv}{D_{KL}\infdivx}

\DeclareMathOperator{\pE}{\tilde{\mathbb{E}}}

\title{
  SoS Degree Reduction with Applications to Clustering and Robust Moment Estimation\thanks{This project has received funding from the European Research Council (ERC) under the European Union’s Horizon 2020 research and innovation programme (grant agreement No 815464)
}
}

\author{
  David Steurer\thanks{ETH Z\"urich.}
  \and
  Stefan Tiegel\footnotemark[2]
}

\begin{document}

\pagestyle{empty}


\maketitle
\thispagestyle{empty} 


\begin{abstract}

  We develop a general framework to significantly reduce the degree of sum-of-squares proofs by introducing new variables.
To illustrate the power of this framework, we use it to speed up previous algorithms based on sum-of-squares for two important estimation problems, clustering and robust moment estimation.
The resulting algorithms offer the same statistical guarantees as the previous best algorithms but have significantly faster running times.
Roughly speaking, given a sample of $n$ points in dimension $d$, our algorithms can exploit order-$\ell$ moments in time $d^{O(\ell)}\cdot n^{O(1)}$, whereas a naive implementation requires time $(d\cdot n)^{O(\ell)}$.
Since for the aforementioned applications, the typical sample size is $d^{\Theta(\ell)}$, our framework improves running times from $d^{O(\ell^2)}$ to $d^{O(\ell)}$.

\end{abstract}

\clearpage


\microtypesetup{protrusion=false}
\tableofcontents{}
\microtypesetup{protrusion=true}

\clearpage

\pagestyle{plain}
\setcounter{page}{1}


\section{Introduction}
\label{sec:introduction}

The \emph{Sum-of-Squares} hierarchy is a hierarchy of semidefinite programs which has proven to be a powerful tool in the theory of approximation algorithms~\cite{GW_Max_Cut, ARV_sparsest_cut}.
More recently it has also given rise to a flurry of algorithms for estimation problems such as various tensor~\cite{SoS_for_dictionary_learning_and_tensor_decomposition, SoS_tensor_1, SoS_tensor_2, SoS_tensor_3}, clustering~\cite{Hopkins_Li_clustering}, and robust estimation problems~\cite{Kothari_Steinhardt_Steurer_robust_moment_estimation, list_decodable_regression, Kothari_robust_regression}, often yielding significant improvements over existing algorithms and in some cases even the first efficient ones.
The hierarchy is based on the \emph{sum-of-squares proof system} which on a high-level allows to argue about non-negativity of polynomials by manipulating a set of polynomial inequalities.
Most importantly, it can be algorithmically exploited in the sense that certain proofs in this proof system directly certify approximation guarantees of algorithms based on the hierarchy.
The running time of these algorithms depends mainly on the number of variables involved and the maximum degree of the polynomials occurring in the inequalities mentioned above.
In general using a higher degree often leads to more accurate solutions but also requires more time.
In this work, we show how we can significantly reduce the degree of a wide range of sum-of-squares proofs in an almost black-box manner while still certifying similar guarantees and thus giving a direct speed-up for concrete algorithms.
As two examples we will consider estimation algorithms for \emph{clustering} and \emph{outlier-robust moment estimation}.
We hope that this technique can inform future algorithms based on the sum-of-squares hierarchy.

\paragraph{Sum-of-squares and estimation problems}

The main idea behind the works listed above is to turn so-called \emph{identifiability proofs} into efficient algorithms.
In (parametric) estimation problems we are typically given samples $y_1, \ldots, y_n$ from a distribution $\cD_{\theta^*}$ parameterized by some $\theta^*$.
Our goal is to recover an estimate $\hat{\theta}$ that is close to the real $\theta^*$, having only access to $y_1, \ldots, y_n$.
An identifiability proof is then a proof of the fact that this is information-theoretically possible, i.e., that there cannot be two different $\theta$ that lie far apart but generate similar samples with some substantial probability.
As long as this identifiability proof can be expressed by a low-degree sum-of-squares proof, there also exist an efficient algorithm for computing $\hat{\theta}$ close to $\theta^*$ (see the algorithmic meta-theorem in the survey \cite{sos_estimation_survey}).
More precisely, the algorithm runs in time $N^{O(\ell)}$, where $N$ is the number of variables of the sum-of-squares proof and $\ell$ is its degree.
The main contribution of this work is a more refined running time analysis that exploits additional structure in the sum-of-squares identifiability proof.
For several applications, namely clustering and robust moment estimation, we show that these structured sum-of-squares proofs exist and, in this way, we derive more efficient algorithms.

\paragraph{Clustering}

In the clustering problem we consider, we are given samples $y_1, \ldots, y_n$ from a uniform mixture of $k$ Gaussian Distributions $\cN(\mu^*_1, I_d), \ldots, \cN(\mu^*_k, I_d)$ and try to recover $\mu^*_1, \ldots, \mu^*_k$.
\footnote{In fact, our algorithms work even if the components are not Gaussian but only satisfy a particular kind of low-order moment bound, which include for example also products of sub-gaussian distributions and rotations thereof. We give a rigorous definition later.}
Further, we denote by $k^*_{r}$ the component $y_r$ was generated from for each $r \in  \brac{n}$ and define the $j$-th \emph{cluster} as $S_j^* = \Set{r \suchthat k^*_r = j}$ for $j = 1, \ldots, k$.
We also strive to find (up to a permutation) accurate estimates of these clusters.
To make this problem information-theoretically possible to solve, we assume a lower bound on the  \emph{minimum separation} of the means, defined as $\Delta \coloneqq \min_{i, j \in \brac{n}: i \neq j} \normt{\mu^*_i - \mu^*_j}$.
A natural question in this context is what the smallest minimum separation is such that we can still hope to (efficiently) find a good approximation.
For more than fifteen years this barrier lay at $\Omega (k^{1/4})$, below nothing was known.
Above this, an algorithm by Vempala and Wang~\cite{single_linkage_clustering} projecting all samples in a suitable subspace and then performing a simple distance based clustering solved the problem efficiently.
In~\cite{Regev_lower_bound} Regev and Vijayaraghavan showed that it is information-theoretically impossible to recover the means and clusters in sub-exponential time when $\Delta = o(\sqrt{\log k})$.
Recently, new polynomial time algorithms emerge when $\Delta = \cO(k^\e)$ for any fixed constant $\e>0$~\cite{Hopkins_Li_clustering,Diakonikolas_robust_estimation_and_clustering,Kothari_Steinhardt_Steurer_robust_moment_estimation}.
(Furthermore, these works also give quasi-polynomial time algorithms for the regime just above the threshold determined by Regev et al., namely for $\Delta = \cO(\sqrt{\log k})$.)

\paragraph{Robust-moment estimation}

In this problem, we are given samples $x_1, \ldots, x_n$ from some distribution $\cD$ over $\R^d$ and wish to find estimates for its mean and low-degree moments  $\Esymb_{x \in \cD} x^{\otimes l}$ for all $0 \leq l \leq k$ fro some $k$.
However, up to an $\eps$-fraction of the samples might have been arbitrarily corrupted by an adversary - these are the outliers.
Algorithms which still succeed in this setting are referred to as \emph{outlier-robust} or short robust.
For the most part finding robust algorithms even for more basic problems such as estimating the mean and covariance matrix  remained computationally intractable in the high-dimensional setting, i.e., the dependence on the dimension $d$ was exponential.
Over the past few years the first works emerged giving the first computationally efficient (robust) estimators for these restricted problems, albeit making strong distributional assumptions (e.g., that $\cD$ is a Gaussian Distribution or a mixture thereof) or having an estimation error depending on the underlying dimension $d$~\cite{diakonikolas_robust, diakonikolas_robust_gaussian, lai2016agnostic}.
Recently, Kothari et al.~\cite{Kothari_Steinhardt_Steurer_robust_moment_estimation} gave the first algorithm robustly estimating also higher-order moments with dimension-independent error and being applicable to a much larger class of distributions.
Further, they give matching information theoretical lower bounds for the class of distributions considered.

We also remark on some similar problems:
A more recent line of work considers the problem of regression in the presence of outliers~\cite{charikar2017learning, Kothari_robust_regression}.
Further, a related but conceptually different setting is the one of so-called \emph{list-decodable learning}.
Here, we consider also the case where $\eps$ is large, i.e., in particular greater than $1/2$ and we search to give a list of possible candidates, say mean vectors or covariances matrices, that is as small as possible and with the guarantee, that at least one of the entries is close to the true parameters of $\cD$.
For recent developments in this area see~\cite{Diakonikolas_robust_estimation_and_clustering, charikar2017learning, list_decodable_1, list_decodable_regression}.

Turning back to the problem of robustly estimating low-degree moments, we note the following:
Without further restrictions on $\cD$ the situation might be hopeless, Just consider the case where $\cD$ outputs a constant with probability $1 - \e$ and another non-constant random variable with probability $\e$.
It turns out that requiring that the low-degree moments of the distribution behave "subgaussian-like" is enough.
More precisely, Kothari et al.~\cite{Kothari_Steinhardt_Steurer_robust_moment_estimation} require the following:
\begin{definition}[Certifiable subgaussianity/Certifiable hypercontractivity,~\cite{Kothari_Steinhardt_Steurer_robust_moment_estimation}]
\label{def:cert_subgaussian}
A distribution $\cD$ over $\R^d$ with mean $\mu$ is called \emph{$t$-certifiably $C$-subgaussian} or certifiably $(C,t)$-hypercontractive for $C > 0$ if for every positive integer $t' \leq t/2$ there exists a \emph{degree-$t$ sum-of-squares proof}~\footnote{We will formally define sum-of-squares proofs later. In this case we require that we can write $\Paren{C t' \cdot \Esymb_{\cD} \iprod{x-\mu, u}^2}^{t'} - \Esymb_{\cD} \iprod{x-\mu, u}^{2t'} = r_1(u)^2 + \ldots + r_k(u)^2$ for some $k$ and $r_1, \ldots, r_k$ polynomials in $u$ of degree at most $t/2$.} of the degree-$2t'$ polynomial inequality:
\begin{align*}
\forall u \in \R^d: \Esymb_{\cD} \iprod{x-\mu, u}^{2t'} \leq \Paren{C t' \cdot \Esymb_{\cD} \iprod{x-\mu, u}^2}^{t'}
\end{align*}
\end{definition}
Beware, that for us it is enough to constrain the low-degree moments of the distribution whereas normally subgaussianity constraints all of the moments of a distribution.
Hence, this definition is much less restrictive.
Kothari et al. show that a large class of distributions satisfy the above definition.
Namely, all distributions satisfying the Poincar\'e-Inequality, which not only includes products of sub-gaussian distributions but also strongly log-concave distributions and Lipschitz continuous transformations of such.

\paragraph{Previous work on speeding up algorithms based on sum-of-squares}

We remark that there exist previous work which gives speed-ups for existing algorithms based on sum-of-squares, albeit quite different from ours.
First, the work of Guruswami and Sinop~\cite{SoS_speed_up_FOCS_2012} exploits the fact that for some problems the rounding step of sum-of-squares algorithms only relies on a small part of the solution of the underlying semidefinite program and not all of it.
Hence, they propose an alternative algorithm computing only the part needed.
Roughly speaking, for the rounding algorithms they consider they are able to match the guarantees of the $r$-th level of the sum-of-squares hierarchy but require only time $n^{\cO(1)}2^{\cO(r)}$ whereas a naive implementation requires time $n^{\Omega(r)}$.
We remark that our techniques differ substantially from theirs in the sense that ours are independent of any rounding scheme.
Further, we do not need to adapt the sum-of-squares algorithm but can continue to use its vanilla version.

Second, and more on the estimation algorithms side, there has been a line of work extracting fast spectral algorithms from sum-of-squares proofs~\cite{fast_SoS_planted, fast_SoS_tensor, fast_SoS_tensor_robust}.
These algorithms exploit the same high-degree information as sum-of-squares relaxations without computing a large semidefinite program.
This is usually done by carefully analyzing the behavior of already existing sum-of-squares algorithms and then tailoring a spectral algorithm towards this specific instance.
Although for the problems at hand they often give significant speed-ups which are also of practical relevance they are problem-specific.

\subsection{Results}

In the following, we present our main results.
Since we need to give a more rigorous introduction to the SoS proof system before being able to give a precise version our meta-theorem we only state an informal one at the end of the next section and defer the full version to Section~\ref{sec:linearization}.
However, we can already state our improved guarantees for the Clustering and Robust-moment Estimation Problem.

For the former, we achieve the following:
\begin{theorem}
\label{thm:clustering_final}
For every even $t$ there exists an algorithm that given $n = k^{\cO(1)}d^{\cO(t)}$ samples from a $d$ dimensional uniform Gaussian Mixture distribution with $k$ components, means $\mu^*_1, \ldots, \mu^*_k$, covariances $\Sigma^*_1 = \ldots = \Sigma^*_k = I_d$, and minimum separation $\Delta \geq \cO(\sqrt{k^{1/t}t})$ finds a clustering $\hat{S}_1, \ldots, \hat{S}_k$ of the points and means $\hat{\mu}_1, \ldots, \hat{\mu}_k$ in time $n^{\cO(1)}$ such that with probability $1 - \frac{1}{\mathrm{poly}(k)}$ for all $i \in [k]$ we have that $\hat{S}_i$ contains at least a $1- \frac{1}{\mathrm{poly}(k)}$ fraction of the points in $S_i^*$ and that $\normt{\mu^*_i - \hat{\mu}_i} \leq \frac{1}{\mathrm{poly}(k)}$.
\end{theorem}

The statistical guarantees of this algorithm match the previous best results for mixtures of spherical Gaussians \cite{Hopkins_Li_clustering,Kothari_Steinhardt_Steurer_robust_moment_estimation,Diakonikolas_robust_estimation_and_clustering}.
However, we improve previous running times of sum-of-squares based algorithms from $n^{O(t)}=k^{O(t)}\cdot d^{O(t^2)}$ to $n^{O(1)}=k^{O(1)}\cdot d^{O(t)}$ \cite{Hopkins_Li_clustering,Kothari_Steinhardt_Steurer_robust_moment_estimation}.
We emphasize that the algorithm \cite{Diakonikolas_robust_estimation_and_clustering} (not explicitly based on sum of squares) actually has the same running time as ours but considers only mixtures of spherical Gaussians whereas we show in Section~\ref{sec:clustering} that our results can easily be extended to mixtures of so-called $t$-explicitly bounded distributions.
In particular, every $t$-certifiably subgaussian distribution is also $t$-explicitly bounded.


We also remark that using our estimates as input to the boosting algorithm in \cite{Regev_lower_bound} we can achieve an arbitrary error $\delta$ for each mean estimate using only $\mathrm{poly}(k, d, 1/\delta)$ additional time and samples.
However, also without this additional boost in accuracy we already achieve $\frac{1}{\mathrm{poly}(k)} \ll \Delta$ in both of the above regimes.

Finally, we remark that the sum-of-squares proof underlying our clustering result is significantly simpler than previous ones.
In particular, our key algorithmic ingredient \ref{thm:clustering_subset} doesn't require any kind of entropy maximixation like \cite{Hopkins_Li_clustering} and applies to a wide range of mixture distribution (without any minimum separation requirement).

\medskip

Regarding robust moment estimation we obtain the following results:
\begin{theorem}[Robust mean and covariance estimation under certifiable subgaussianity]
\label{thm:moment_estimation}
For every $C > 0$ and even $t \in \N$ there exists an algorithm that given an $\eps$-corrupted sample $S$ of a $t$-certifiably $C$-subgaussian distribution with mean $\mu^*$ and covariance $\Sigma^*$ such that $\card{S} = n \leq (C + d)^{\cO(t)}$ outputs in time $n^{\cO(1)}$ a mean-estimate $\hat{\mu}$ and covariance-estimate $\hat{\Sigma}$ such that with high probability it holds that \footnote{As a further technical assumption we require that $\lambda_{min}(\Sigma^*) \geq 2^{-n^{\cO(1)}}$, i.e., $n \geq \log \Paren{ \frac{1}{\lambda_{min}(\Sigma^*)} }$, where $\lambda_{min}(\Sigma^*)$ denotes the minimum eigenvalue of $\Sigma^*$}:
\begin{align*}
&\normt{\mu^* - \hat{\mu}} \leq \cO(Ct)^{1/2} \cdot \eps^{1 - 1/t} \cdot \Norm{\Sigma^*}^{1/2} \\
&\normt{(\Sigma^*)^{-1/2} (\mu^* - \hat{\mu}) } \leq  \cO(Ct)^{1/2} \cdot \eps^{1 - 1/t} \\
&(1-\delta) \Sigma^* \sle \hat{\Sigma} \sle (1+\delta) \Sigma^*
\end{align*}
for $\delta \leq  \cO(Ct) \cdot \eps^{1 - 2/t}$.
\end{theorem}
Analogously to the Clustering Problem, we remark that the guarantees we achieve are the same as in~\cite{Kothari_Steinhardt_Steurer_robust_moment_estimation}, but the running time is again $(C+d)^{\cO(t)}$ versus $(C+d)^{\cO(t^2)}$ before.

Similarly, we can give the following guarantees for estimating higher-order moments.
\begin{theorem}[Robust higher-moment estimation]
\label{thm:higher_moment_estimation}
For every $C > 0$ and even $t \in \N$ there exists an algorithm that given an $\eps$-corrupted sample $S$ of a $t$-certifiably $C$-subgaussian distribution with moment-tensors $M^*_2 \in \R^{d^2}, \ldots, M^*_t \in \R^{d^t}$ such that $\card{S} = n \leq (C + d)^{\cO(t)}$\footnote{Again, as a further technical assumption we require that $n \geq \log \Paren{ \frac{1}{\lambda_{min}(M_2^*)} }$} outputs in time $n^{\cO(1)}$ moment estimates $\hat{M}_2 \in \R^{d^2}, \ldots, \hat{M}_t \in \R^{d^t}$ such that with high probability for every $r \leq t/2$ there exist a certificate in form of a degree-$2r$ sum-of-squares proofs (in variables $u$) that $\iprod{M_r^* - \hat{M}_r, u^{\otimes r}}^2 \leq \delta_r^2$ for all $u$ such that $\iprod{u, M_2^*  u}^r = 1$, where $\delta_r \leq  \cO(Ct)^{r/2} \cdot \eps^{1 - r/t}$.
\end{theorem}
As before, the improvement here lies in the running time of $(C+d)^{\cO(t)}$ versus $(C+d)^{\cO(t^2)}$.
Note, that the guarantees of this theorem imply that for all $u \in \R^d$ we have $\iprod{M_r^* - \hat{M}_r, u^{\otimes r}}^2 \leq \delta_r^2 \iprod{u, M_2^*  u}^r$, i.e., for every direction $u$ we can estimate the $r$-th moment in direction $u$ with error comparable to the second moment in this direction.
We phrase the guarantee as above, since it is crucial for the applications listed in~\cite{Kothari_Steinhardt_Steurer_robust_moment_estimation} that we have a sum-of-squares formulation of this fact.
Although listing the specific guarantees of these here would surpass the scope of this paper, we remark that the applications include robust versions for Independent Component Analysis, given observations $Ax$ for a matrix $A$ and a random vector $x$ try to recover the columns of $A$, and learning mixtures of spherical Gaussians.
Although we also study the latter, the assumptions and guarantees given there are rather different from our setting.
Hence, we refer the interested reader to study the original paper.

\section{Techniques}
\label{sec:techniques}

In this section we outline how we can reduce the degree of sum-of-squares (or short SoS) proofs.
To this end, we first give a brief, somewhat simplified, introduction to sum-of-squares proofs but which will already suffice to understand our main techniques.

Let $x \in \R^n$ and $q_1, \ldots, q_m$ be polynomials and consider the set $\cA = \Set{x \suchthat q_1(x) \geq 0, \ldots, q_m(x) \geq 0}$ and another polynomial $p$ over $\R^n$.
Suppose we would like to certify that $p(x) \geq 0$ for all $x \in \cA$.
One way to do this is if we can write $p(x) = b_0(x) + \sum_{i=1}^m b_i(x) q_i(x)$ where $b_0, \ldots, b_m$ are \emph{sum-of-squares} polynomials, i.e., we have for example $b_0(x) = r_1(x)^2 + \ldots r_k(x)^2$ for some $k$ and $r_1, \ldots, r_k$ polynomials.
This decomposition we call a \emph{sum-of-squares proof} and we say that it is of degree $t$, if $\deg(b_0) \leq t$ and $\deg(b_i q_i) \leq t$ for all $i \in \brac{m}$.
If there is a degree $t$ proof that $\cA$ implies $p \geq 0$ we write $\cA \sststile{t}{} p \geq 0$.
In the framework of estimation algorithms $\cA$ would roughly describe our guesses $x$ for the correct parameter and by setting $p = \e - \snormt{x - \theta^*}$ we have that $\cA \sststile{t}{} p \geq 0$ gives us a certificate that in fact all of these are of the desired accuracy.
Following the paradigm of turning identifiability proofs into algorithms, the sum-of-squares proofs typically mimics classical identifiability proofs.

However, finding any $x \in \cA$ is NP-hard in general.
On the other hand, and without going into too much detail, these kinds of certificates can be algorithmically exploited in the sense that if the above SoS proof of degree $t$ exists  we can most of the time find an estimate $\hat{x}$ such that $p(\hat{x}) \geq 0$ as well.
Further, under some mild assumptions this can be done in time $n^{\cO(t)}$, where $n$ again is the number of variables we use.
In general this requires a rounding step which might be problem dependent.
However, \cite{sos_estimation_survey} showed that a large range of estimation problems fit this framework and give a general rounding strategy applicable to these cases.

The main contribution of this work is that we can find substantially more succinct certificates, i.e., SoS proofs with a much lesser degree, at the expense of introducing some additional variables.
Hence, if we could say achieve a constant degree by using only $\mathrm{poly}(n)$ additional variables this would give hope to find algorithms running in time $n^{\cO(1)}$ but having the same approximation guarantees.
We show that under some mild conditions exactly this is possible.
Additionally, our construction works in an almost black-box matter and can be applied to any SoS proof.
It also seems that this linearization scheme is compatible with most standard rounding algorithms in the sense that they do not need to be changed dramatically.

To begin with, consider the following linearization operator $f$ sending polynomials $p$ of degree at most $t$ over $\R^n$ to linear functions over $\R^{(n+1)^t}$, where for simplicity we index vectors of $\R^{(n+1)^t}$ by multi-indices $\alpha$ with $n+1$ entries (0 to $n$) of size exactly $t$.
Then the polynomial $p = \sum_{\card{\alpha} = t} p_\alpha (1,x)^\alpha$ will get mapped to $f (p) = \sum_{\card{\alpha} = t} p_\alpha x_\alpha$.
Further, consider the subset $\cB \sse \R^{(n+1)^t}$ containing all $x$ such that $x_\alpha = \prod_{i=1}^{n} x_{e_i}^{\alpha_i}$, where by $e_i$ we denote the multi-index that has a 1 in position $i$ and 0s everywhere else.
Clearly, on $\cB$ we have that $f(p)$ behaves "just like $p$ on $\R^n$" for any $p$, in the sense that $(f(p))(x) = p((x_{e_1}, \ldots, x_{e_n}))$.
Returning to the settings of SoS proofs, suppose we have $\cA = \Set{q_1 \geq 0, \ldots, q_m \geq 0} \sststile{t}{} p \geq 0$ for some $p$.
Defining $\cA' \coloneqq \Set{f(q_1) \geq 0, \ldots, f(q_m) \geq 0}$ intuitively we would expect that $\cA' \cup \cB \sststile{3}{} f(p) \geq 0$ by replacing all polynomials in the SoS proof by their linearized versions - we will give a rigorous proof in Section~\ref{sec:linearization}.
Further, it also turns out that for most our rounding algorithms this guarantee yields the same results as $\cA \sststile{t}{} p \geq 0$.

Turning back to finding efficient algorithms, these kind of new certificates can roughly speaking be turned into algorithms running in time $((n+1)^t)^{\cO(1)} = n^{\cO(t)}$ and hence we still have not gained anything.
The key observation is now the following:
In many cases the polynomials occurring in the original SoS proof do not contain all monomials of size smaller than $t$, but only very few!
I.e., for most multi-indices $\alpha$ all corresponding coefficients are 0.
Thus, instead of working over $\R^{(n+1)^t}$ we can apply a linearization scheme sending polynomials to linear functions over $\R^N$, where $N$ is the number of monomials that actually occur in the SoS proof, by simply omitting all that do not.
It is easy to see that none of the above arguments change and hence, we can expect a running time of $N^{\cO(1)}$.

To exemplify this, we will give a simplified overview of how we will apply this scheme to the Clustering Problem.
We will extract the most important details to not distract from the main ideas, the full and rigorous version can be found in Section~\ref{sec:clustering}.
However, the heart of all improvements in running time can already be seen here.
In what follows, we will assume for simplicity that we are given $y_1, \ldots, y_n$ samples from a uniform mixture of one-dimensional Gaussians with covariances equal to 1 and means $\mu^*_1, \ldots, \mu^*_k$.
As an intermediate step, we will try to give a set of $n$ mean estimates, one for each sample, with the guarantee that they lie close to the means used to generate the respective sample.
These latter we will denote by $\mu^*_{k_r}$.
In a simple post-processing step, we will be able to extract form them the original clustering and transform them into only $k$ mean estimates.
We hence introduce a system of inequalities in variables $\mu_r$ for $r \in \brac{n}$ (and some other auxiliary variable which we omit here for simplicity).
Some of these constraints will be of constant degree and some of degree $t$, however, they will contain only few terms of non-constant degree.
These higher-degree terms, let us denote them by $q_i$, will be used to enforce a moment-bound constraint similar to the one given in Definition~\ref{def:cert_subgaussian}.
Omitting the low-degree constraints and assuming that the sum-of-squares polynomials in our SoS proof are in fact constants for simplicity, a SoS proof that our mean-variables are close to the real ones would take the form $\e - \frac{1}{n} \sum_{r=1}^n (\mu^*_{k_r} - \mu_r)^2 = \sum_{i=1}^m (c_i)^2 q_i(\mu)$ where $q_i(\mu) = p_i(\mu)  + \sum_{j=1}^k p_{i,j} \mu_j^t$ for some constant degree polynomial $p_i$.
Hence, introducing a new variable $\mu_{j,t}$ for each $j \in \brac{k}$ and the constraints $p_i(\mu)  + \sum_{j=1}^k p_{i,j} \mu_{j,t} \geq 0$ we can easily see that there exists an SoS proof of $\e - \frac{1}{n} \sum_{r=1}^n (\mu^*_{k_r} - \mu_r)^2  \geq 0$ of only constant degree.
Further, we introduced only $k \ll n$ new variables and hence, we can hope for algorithms with running time $n^{\cO(1)}$ instead of $n^{\cO(t)}$.
We remark that in the full version we will introduce $n^{\cO(1)}$ new variables, but this does not change the speed-up we gain.

Most importantly, the proposed linearization scheme goes well beyond the instant-specific description above and in fact, can be applied to any sum-of-squares proof in a black-box manner.
Under some mild assumptions we can now state our main theorem for estimation problems.
In Section~\ref{sec:linearization}, we make all of these explicit and give a more formal version.
\begin{theorem}[Informal]
\label{thm:improved_sos_estimation_informal}
Consider the distribution $\cD_{\theta^*}$ and let $y_1, \ldots, y_n \sim \cD_{\theta^*}$, where $\theta^*$ is the true parameter we wish to estimate.
Suppose there exists a system of polynomial inequalities $\cA_{y_1, \ldots, y_n} = \Set{q_1 \geq 0, \ldots, q_m \geq 0}$ in variables $x \in \R^n$ such that $\cA \sststile{t}{} \e - \snormt{x - \theta^*} \geq 0$ and let $N$ be the number of monomials occurring in the SoS proof.
Then there exists an algorithm running in time $N^{\cO(1)}$ yielding an estimate $\hat{x}$ such that $\snormt{\hat{x} - \theta^*} \leq \e + 2^{-n^{\Theta(1)}}$ \footnote{The additional additive factor of $2^{-n^{\Theta(1)}}$ is unavoidable for reasons becoming clear later. However, in most of our applications this factor will be negligible.}.
\end{theorem}

\section{Preliminaries}
\label{sec:preliminaries}

In this section, we will introduce sum-of-squares proofs and their convex duals, so-called pseudo-distributions.

\subsection{Sum-of-squares proofs}

As outlined in the previous section, we wish to certify the non-negativity of multivariate polynomials using the fact that squares are non-negative.
Again, consider the set $\cA = \Set{q_1 \geq 0, \ldots, q_m \geq 0}$ defined by multivariate polynomials $q_1, \ldots, q_m$ from $\R^n$ to $\R$ and some other polynomial $p$.
A \emph{sum-of-squares proof} that $p$ is non-negative for all $x \in \cA$ is now a set of sum-of-squares polynomials $b_S = \sum_{j=1}^{n_S} u_j^2$ for $S \sse \brac{m}$ such that
\begin{align}
p = \sum_{S \subseteq \brac{m}} b_S \prod_{i \in S} q_i \label{sos_eq1}
\end{align}
Each summand is clearly non-negative as the $b_S$ are sum-of-squares and the product of the $q_i$ is as well by our assumption that the constraints in $\cA$ are satisfied.
Further, if $\deg \Paren{b_S \prod_{i \in S} q_i} \leq t$ for all $S \subseteq \brac{m}$ we say that the sum-of-squares proof has \emph{degree t} and write $\cA \sststile{t}{x} p \geq 0$ where $x$ denotes that the proof uses the variables $x$.
Their degree will be our measure of complexity for sum-of-squares proofs.
In Section~\ref{sec:techniques} we considered the special case where we force $b_S \equiv 0$ if $\card{S} > 1$.
We can also achieve this form by introducing the constraint $( \prod_{i \in S} q_i ) \geq 0$ for each $S \sse \brac{m}$ such that $\card{S} > 1$ and $b_S \not\equiv 0$.
We denote the number of constraints additionally introduced like this by $M$.

To make our lives a bit easier we introduce some more notation:
We write $\cA \sststile{t}{x} p \geq p'$ if $\cA \sststile{t}{x} p - p' \geq 0$ and $\cA \sststile{t}{x} p = p'$ if $\cA \sststile{t}{x} p \leq p'$ and $\cA \sststile{t}{x} p \geq p'$.
An easy observation shows that sum-of-squares proofs are stable under compositions which allows us to build up complex sum-of-squares from simpler ones just as with regular equations.
More precisely, if $\cA \sststile{t}{x} p \geq p'$ and $\cA \sststile{t'}{x} p' \geq p''$, then it also holds that $\cA \sststile{t''}{x} p \geq p''$, where $t'' = \max\{t, t'\}$.
Further, if $\cA \sststile{t}{} \cB$ and $\cB \sststile{t'}{} p \geq 0$ we also have $\cA \sststile{t \cdot t'}{} p \geq 0$.

\subsection{Pseudo-expectations}

We next turn to the convex duals of SoS proofs, so-called \emph{pseudo-distributions}.
Although more results on duality exist, we only introduce what is necessary for us in this work.
For a more thorough exposition we refer the reader to~\cite{SoS_online_tutorial} and~\cite{SoS_general_survey} and the references therein.
The definition is as follows:
\begin{definition}[Pseudo-distribution]
\label{def:pE}
For every $d \in \N$, a \emph{degree-$d$ pseudo-distribution} is a finitely-supported function $\mu$ from $\R^n$ to $\R$ such that $\sum_{x \in \text{supp($\mu$)}} \mu(x) = 1$ and $\sum_{x \in \text{supp($\mu$)}} \mu(x) f^2(x) \geq 0$ for all polynomials $f$ of degree at most $d/2$.
\end{definition}
Where we define $\text{supp($\mu$)} \coloneqq \Set{x \suchthat \mu(x) \neq 0}$ as the support of $\mu$.
We can see, that in contrast with regular distributions we don't require the function $\mu$ to be non-negative.
The \emph{pseudo-expectation} of a function under $\mu$ is then defined as $\pE_\mu f \coloneqq \sum_{x \in \text{supp($\mu$)}} \mu(x) f(x)$.
Further, we define the following:
\begin{definition}[Constrained pseudo-distribution]
\label{def:constrained_pE}
Let $\mu$ be a degree-$d$ pseudo-distribution and $\cA = \Set{q_1 \geq 0, \ldots, q_m \geq 0}$ a system of polynomial inequalities.
We say that $\mu$ \emph{satisfies} $\cA$ if for all $S \sse \brac{m}$ and all sum-of-squares polynomials $h$ such that $\deg \Paren{ h \cdot \prod_{i \in S} q_i } \leq d$ we have
\begin{align*}
\pE_\mu h \cdot \prod_{i \in S} q_i \geq 0
\end{align*}
Further, we say that $\mu$ satisfies $\cA$ approximately if $\pE_\mu h \cdot \prod_{i \in S} q_i \geq -\eta \normt{h} \prod_{i \in S} \normt{q_i}$, where by $\normt{p}$ for a polynomial $p$ we denote the Euclidean Norm of its coefficient vector and $\eta = 2^{-n^{\Omega(d)}}$.
\end{definition}

Next, we state the key relation between pseudo-distributions and SoS proofs.
\begin{fact}
\label{fact:pE_and_sos_proofs}
Let $\mu$ be a degree-$d$ pseudo-distribution that satisfies $\cA$ and suppose that $\cA \sststile{t}{} p \geq 0$.
Let $h$ be an arbitrary sum-of-squares polynomial, if $\deg h + t \leq d$ we have $\pE_\mu h \cdot p \geq 0$. 
In particular, we have $\pE_\mu p \geq 0$ as long as $t \leq d$.
If $\mu$ only approximately satisfies $\cA$ it holds that $\pE_\mu p \geq -\eta \normt{h}$ for $\eta = 2^{-n^{\Omega(d)}}$.
\footnote{A further technical assumption for the latter case is that the sum-of-squares proof has bit-complexity at most $n^{\cO(d)}$. In all the cases we consider this will be satisfied and hence, we will not mention it again. We refer to~\cite{odonnell_bit_complexity} and~\cite{raghavendra_bit_complexity} for more details on why this is necessary.}
\end{fact}

An essential fact is that pseudo-distributions approximately satisfying a system of polynomial inequalities can be found efficiently as long as the constraints have bit-complexity at most $(m+n)^{\cO(1)}$.
\begin{fact} 
\label{fact:pE_efficient_optimization}
Given any feasible system of polynomial constraints $\cA = \Set{q_1 \geq 0, \ldots, q_m \geq 0}$ \footnote{For technical reason we also assume that $\cA$ is \emph{explicitly bounded} meaning it contains a constraint of the form $\snormt{x} \leq B$ for some large number $B$. Usually we can even take this to be at least polynomial in our variables. We remark that for all the problems we are considering this can be added without changing our proofs. We will give remarks justifying this in the respective places.} over $\R^n$ whose bit-complexity satisfies the constraints mentioned above and $\eta = 2^{-n^{\Theta(d)}}$ we can in time $(n+m)^{\cO (d)}$ find a degree-$d$ pseudo-distribution that satisfies $\cA$ up to error $\eta$. \footnote{The positivity and normalization constraint of definition~\ref{def:pE} can be satisfied exactly however.}
\end{fact}

Pseudo-expectations satisfy many natural properties that we are used to from regular expectations.
One particularly useful one is the Cauchy-Schwarz Inequality.
\begin{fact}
\label{fact:pE_cauchy_schwarz}
Let $\mu$ a degree-$d$ pseudo-distribution and $p$ and $q$ two polynomials of degree at most $d/2$ then it holds that $\Paren{ \pE_\mu p q }^2 \leq  \pE_\mu p^2 \pE_\mu q^2 $.
Hence, choosing $q \equiv 1$, it in particular holds that $\Paren{ \pE_\mu p }^2 \leq \pE_\mu p^2$ as long as $\deg p \leq d/2$.
\end{fact}

\section{Degree reduction of sum-of-squares proofs via linearization}
\label{sec:linearization}

In this section, we show how we can reduce the degree of sum-of-squares proofs by introducing new variables.
Further, we will state a more formal version of Theorem~\ref{thm:improved_sos_estimation_informal} and give its proof.
The main idea we will use is the following:
If we only have few variables which appear with a high degree in our proof we can substitute these high powers by new variables and introduce new constraints which ensure that they behave as expected.
As a consequence, we obtain a constant degree proof of essentially the same statement.
Recall, that for a sum-of-squares proof of the form $p = \sum_{S \sse \brac{m}} b_S \prod_{i \in S} q_i$ we define $M = \card{\Set{S \sse \brac{m} \suchthat \card{S} > 1, b_S \not\equiv 0}}$.
We remark that in most cases we have $M = m^{\cO(1)}$.
We will proof the following theorem:
\begin{theorem}
\label{thm:linearization}
Let $p \colon \R^n \to \R$ be a multivariate polynomial of degree at most $t$.
Suppose there exists a system of polynomial inequalities $\cA = \Set{q_1 \geq 0, \ldots, q_m \geq 0} \sse \R^n$ for $q_1, \ldots q_m \colon \R^n \to \R$ such that $\cA \sststile{t}{x} \Set{p(x) \geq 0}$.
Then there exists $n' \geq n, m' \geq m$, $f \colon \R^n \to \R^{n'}$ injective and $\cA' = \Set{q'_1 \geq 0, \ldots, q'_{m'} \geq 0} \sse \R^{n'}$ for $q'_1, \ldots q'_{m'} \colon \R^{n'} \to \R$ and $p' \colon \R^{n'} \to \R$ linear functions such that the following holds:
\begin{enumerate}
\item $\forall x \in \R^{n}, i \in \brac{n} \colon (f(x))_i = x_i$ \label{linearization_1}
\item $\forall x \in \R^n \colon x \in \cA \iff f(x) \in \cA'$ \label{linearization_2}
\item $\forall x \in \R^n \colon p(x) = p'(f(x))$ \label{linearization_3}
\item $\cA' \sststile{3}{y} \Set{p'(y) \geq 0}$ \label{linearization_4}
\end{enumerate}

Further, let $n_1, \ldots, n_t$ denote the number of degree $1, \ldots, t$ terms that occur in the first sum-of-squares proof respectively and define $n_{\mathrm{max}} \coloneqq \max_{1 \leq i \leq t} n_i$.
Then we have $n' = n + \cO((t \cdot n_{\mathrm{max}})^3) = \cO((t \cdot n_{\mathrm{max}})^3)$ and $m' = m + M + \cO((t \cdot n_{\mathrm{max}})^3)$.
\end{theorem}

Since its proof is purely algebraic we defer it to Section~\ref{sec:linearization_proof}.
After adding new variables and constraints as outlined above, the basic idea is to construct a "linearized version" of all the polynomials involved as alluded to in Section~\ref{sec:techniques}.
We instead turn to proving our main theorem, which now can be stated as follows:
\begin{theorem}
\label{thm:improved_sos_estimation}
Consider the distribution $\cD_{\theta^*}$ and let $y_1, \ldots, y_n \sim \cD_{\theta^*}$, where $\theta^*$ is the true parameter we wish to estimate.
Suppose there exists $t \in \N$ and a system of polynomial inequalities $\cA_{y_1, \ldots, y_n} = \Set{q_1 \geq 0, \ldots, q_m \geq 0}$ in variables $x$ such that $\cA \sststile{t}{x} \Snormt{\theta^* - x} \leq \e$.
\footnote{We assume that it has bit-complexity at most polynomial in $n$ ($n^{\cO(t)}$ is fine too).}
Further, let $n_{\mathrm{max}}$ and $M$ be defined analogously as in Theorem~\ref{thm:linearization}.
Then there exists an algorithm running in time $\Paren{ m + M + (t \cdot n_{\mathrm{max}})^3}^{\cO(1)}$ that returns an estimate $\hat{x}$ such that $\Snormt{\theta^* - \hat{x}} \leq \e + 2^{-\Omega(t \cdot n_{\mathrm{max}})}$.
If $m = \mathrm{poly}(t, n_{\mathrm{max}})$, and $M = \mathrm{poly}(t, n_{\mathrm{max}})$ the running time becomes $\mathrm{poly}(t \cdot n_{\mathrm{max}})$.
\end{theorem}
\begin{proof}
Let $p(x) = \e - \Snormt{\theta^* - x}$, then we have $\cA \sststile{t}{x} \Set{p(x) \geq 0}$.
Let $f, \cA', $ and $p'$ be the objects given by the conclusion of Theorem~\ref{thm:linearization} and let $\mu$ be a degree-$3$ pseudo-distribution over $\R^{n'}$ that satisfies $\cA'$ up to error $2^{-\Omega(n')}$, for $n'$ as in Theorem~\ref{thm:linearization}.
We know that $n' = \Omega(t \cdot n_{\mathrm{max}})$.
By Fact~\ref{fact:pE_efficient_optimization} this can be found in time $\Paren{n + m + M + (t \cdot n_{\mathrm{max}})^3}^{\cO(1)} = \Paren{m + M + (t \cdot n_{\mathrm{max}})^3}^{\cO(1)}$.
Define $\hat{x} \in \R^n$ by setting $\hat{x}_i = \pE_\mu x_i = \pE_\mu (f(x))_i $ for $i \in \brac{n}$.
We get:
\begin{align*}
&\e - \Snormt{\theta^* - \hat{x}} = \e - \sum_{i=1}^n \Paren{\pE_\mu \theta^*_i - x_i}^2 \geq \e - \sum_{i=1}^n \pE_\mu \Paren{\theta^*_i - x_i}^2 = \e - \pE_\mu \Snormt{\theta^* - x} \\
&= \pE_\mu p(x) = \pE_\mu p'(f(x)) \geq -2^{-\Omega(n')}
\end{align*}
where the first inequality follows by Cauchy-Schwarz for pseudo-expectations (cf. Fact~\ref{fact:pE_cauchy_schwarz}) and the last equality from Point~\ref{linearization_3} in Theorem~\ref{thm:linearization}.
The last inequality is due to the fact that $\mu$ satisfies $\cA'$ approximately and is of degree 3.
\end{proof}

We pause a second to discuss the running time guarantees of our new theorem and to compare them to what is given by Raghavendra et al.~\cite{sos_estimation_survey}:
Typically have $m = \mathrm{poly}(n)$ and $M = \mathrm{poly}(m)$ since often our sum-of-squares proof will only multiply at most a constant number of constraints.
Hence, our new running time becomes $\mathrm{poly}(t \cdot n_{\mathrm{max}})$, whereas previously we had only $\mathrm{poly}(n^t)$.
Hence, we can see that the exponential dependence on $t$ was replaced by a multiplicative one and the dominating factor now becomes $n_{\mathrm{max}}$.
In the worst case, we have $n_{\mathrm{max}} = n_t = n^t$ and thus we have gained nothing.
However, in all other cases we are strictly improving.
Concretely, as soon as $n_{\mathrm{max}} \ll n^t$ we also have $\mathrm{poly}(t \cdot n_{\mathrm{max}}) \ll \mathrm{poly}(n^t)$.
Where by $n_{\mathrm{max}} \ll n^t$ we are assuming that the difference is large enough to swamp the difference which might be introduced by the different powers hidden in the $\mathrm{poly}(\cdot)$ term.
These could be made explicit, but in any case are very small.
We are thus proposing a new analysis offering a more fine-grained control over the running-time.
In some cases, including the two examples we will investigate in the rest of this paper, we even have $n_{\mathrm{max}} = \mathrm{poly}(n)$ and hence we get a running time of $\mathrm{poly}(t,n)$ which drastically improves over the previous guarantee!

\section{Clustering}
\label{sec:clustering}

In this section, we aim to proof Theorem~\ref{thm:clustering_final}.
We also remark that the theorem continues to hold in the setting where each component of the mixture is not Gaussian but only $t$-explicitly bounded as introduced by~\cite{Hopkins_Li_clustering}.
The formal definition is as follows:
\begin{definition}[$t$-explicitly Bounded Distribution,~\cite{Hopkins_Li_clustering}]
\label{def:t_explicitly_bounded}
Let $t \in \N$.
A distribution $\cD$ over $\R^d$ with mean $\mu$ is called \emph{$t$-explicitly bounded} with variance proxy $\sigma$ if for each even integer $s$ such that $1 \leq s \leq t$ the following equation has a degree $s$ sum-of-squares proof (in variables $u$):
\begin{align*}
\sststile{s}{u} \Esymb_{X \sim \cD} \iprod{X - \mu, u}^2 \leq (\sigma s)^{s/2} \normt{u}^s
\end{align*}
\end{definition}
Note, that any $t$-certifiably $C$-subgaussian distribution is $t$-explicitly bounded with variance proxy $C \cdot \sigma_{\mathrm{max}}$, where $\sigma_{\mathrm{max}}$ is the largest singular value of the covariance matrix of $\cD$.
\cite{Kothari_Steinhardt_Steurer_robust_moment_estimation} shows that a large class of distributions satisfy are $t$-certifiably $C$-subgaussian and hence also $t$-explicitly bounded.
Namely, all distributions satisfying the Poincar\'e-Inequality, which not only includes products of sub-gaussian distributions but also strongly log-concave distributions and Lipschitz continuous transformations of such.

\subsection{Obtaining intermediate estimates}

We will first solve a slightly different problem, which in turn will then allow us to obtain our final clustering.
Recall, that we are given samples $y_1, \ldots, y_n$ and the real means of the mixture components are denoted by $\mu^*_1, \ldots, \mu^*_k$.
For each $r \in \brac{n}$ let $k^*_r$ denote the index of the component the sample $y_r$ was generated from.
Our samples then can be written as $y_r = \mu^*_{k^*_r} + w_r^*$ where $w^*_r$ is a $d$-dimensional standard Gaussian.
Using matrix notation, we can formulate this more compactly as follows:
\begin{align*}
Y = X^* + W^*
\end{align*}
where $Y, X^*$ and $W^* \in \R^{d \times n}$ have as columns the $y_r$s, $\mu^*_r$s, and $w_r^*$s respectively.
The intermediate problem we consider can now be stated as follows:
For each sample $y_r$ we wish to find a mean estimate $\hat{\mu}_r$ that ideally is close to $\mu^*_{k^*_r}$.
Grouping all our mean estimates $\hat{\mu}_r$ in a matrix $\hat{X} \in \R^{d \times n}$, the error metric we will seek to minimize then can be expressed as:
\begin{align*}
\frac{1}{n} \Snorm{X^* - \hat{X}}_F = \frac{1}{n} \sum_{r = 1}^n \Snormt{\mu^*_{k^*_r} - \hat{\mu}_r}
\end{align*}
To give some intuition why this might be interesting consider the following argument:
If the above error is small our estimates $\hat{\mu}_r$ lie close to one of the real means on average.
Also, if it is small enough in comparison with the minimum distance $\Delta$, they are close enough to each other so that we can apply a simple distance based clustering to recover the clusters with high accuracy.
This idea will be made rigorous in Section~\ref{sub_sec:clustering_final}.
First, we will show the following theorem:
\begin{theorem}
\label{thm:clustering_subset}
For every even $t$ there exists an algorithm that given $n = k^{\cO(1)}d^{\cO(t)}$ samples from a $k$-component mixture of $t$-explicitly bounded distributions (with variance proxy $1$) in dimension $d$ finds an assignment of samples to clusters in time $n^{\cO(1)}$ such that $\frac{1}{\card{S}} \sum_{r \in S} \snormt{\mu_{k^*_r}^* - \hat{\mu}_{k_r}} \leq 8 k^{(4+2c)/t} \cdot t$ for all $S \subseteq \brac{n}$ such that $\card{S} \geq \frac{n}{k^c}$ and $c \in \N_{\geq 0}$.
\end{theorem}
Choosing $c = 0$ and $S = \brac{n}$ we recover the setting described above.

We will prove this using Theorem~\ref{thm:improved_sos_estimation} and hence we search for a system of polynomial inequalities $\cA$ and a SoS proof that $\cA \sststile{\cO(t)}{} \frac{1}{\card{S}} \sum_{r \in S} \snormt{\mu_{k^*_r}^* - \mu_{k_r}} \leq 8 k^{(4+2c)/t} \cdot t$.
The difference between a degree-$t$ and degree-$\cO(t)$ proof here is negligible and does not change any of the arguments.
Here we also remark that the error of $2^{-\Omega(n)}$ introduced by Theorem~\ref{thm:improved_sos_estimation} is negligible in comparison to $8 k^{(4+2c)/t} \cdot t = \Omega(\log k) = \omega(1)$.
Also it will be easy to check that we have $m = \mathrm{poly}(n)$ and $M = \mathrm{poly}(m)$.
We will verify at the end that we also have $n_{\text{max}} = \mathrm{poly}(n)$ and hence we are guaranteed that we can compute our estimates in time $n^{\cO(1)}$.
We also remark that although our SoS proof will be different for each $S \sse \brac{n}$ the estimates we compute are the same for all of them; the SoS proofs should be seen as a certificate that they yield a good approximation for various $S$.

The system $\cA$ we introduce tries to capture the Clustering Problem and most notably the fact that the uniform distribution over every cluster is $t$-explicitly bounded.
We remark that this is essentially the same as in~\cite{Hopkins_Li_clustering} but our sum-of-squares proof differs substantially from theirs.
In their paper they also discuss feasibility \footnote{They also show that we can include the constraint $\snormt{x} \leq n^{\cO(1)}$ without changing any arguments and thus making the system $\cA$ explicitly bounded as required by Fact~\ref{fact:pE_efficient_optimization}}.
Let $y_1, \ldots, y_n$ be the samples our algorithm is given and consider the following system of polynomial inequalities in scalar-valued variables $z_r^j$ and vector-valued variables $\mu_j$:
\begin{align*}
\cA \coloneqq	\left\{ 
					\begin{array}{l}
						\forall r \in [n] \sum_{j = 1}^k z_r^j = 1 \\
						\forall j \in [k] \sum_{r = 1}^n z_r^j = \frac{n}{k} \\
						\forall r \in [n], i,j \in [k], i \neq j: z_r^i \cdot z_r^j = 0 \\
						\forall j \in [k], r \in [n]: z_r^j = (z_r^j)^2 \\
						\forall j \in [k]: \forall u \in \R^d: \frac{k}{n} \sum_{r = 1}^n z_r^j \iprod{ y_r - \mu_{k_r}, u }^t \leq 2t^{t/2} \norm{u}_2^t
					\end{array}
				\right\}
\end{align*}
where we define $\mu_{k_r} \coloneqq \sum_{j=1}^k z_r^j \mu_j$.

Intuitively, we can think of the $z_r^j$ variables as assigning each sample to a cluster, i.e., $z_r^j = 1$ if we assign the $r$-th point to the $j$-th cluster and $0$ otherwise.
The above system encodes an assignment of samples to clusters such that they behave \emph{Gaussian like} which is formalised by the moment constraint (last line of $\cA$).
Note, that this last constraint doesn't fit the model we introduced in Section~\ref{sec:preliminaries}, however, later we will show how to phrase it in a way that does.
Since this is a technicality and distracts from the essence of the proof, we will defer this discussion to Section~\ref{sec:toolkit}.
Essentially, we will not only require that it be true for all $u$ but also that this is certified by an SoS proof (in variables $u$).

To give a proof which is as simple as possible, we first make the following observations:
Consider any solution to the system $\cA$.
We now refine the partition of $[n]$ given by our predicted clusters $S_i = \Set{y_j \vert z_r^i = 1}$ by intersecting them with the real clusters $S^*_j$. For $r,r' \in S_i \cap S^*_j$ we observe that $\mu_{k_r} = \mu_{k_{r'}} = \mu_i$ and hence
\begin{align*}
&y_r = \mu_i + w_r = \mu^*_j + w^*_r \\
&y_{r'} = \mu_i + w_{r'} = \mu^*_j + w^*_{r'}
\end{align*}
where we define $w_r \coloneqq y_r - \mu_{k_r}$ and $w_{r'} \coloneqq y_{r'} - \mu_{k_{r'}}$.
Hence, we have
\begin{align*}
\mu_i - \mu^*_j = w^*_r - w_r = w^*_{r'} - w_{r'}
\end{align*}

Since this is the essence of our later proofs, we would like to introduce the variables $v_r \coloneqq w^*_{r} - w_r = \mu_i - \mu_j^*$ for $r \in [n]$.
The above equation tells us that the $v_r$ are constant on each of the sets $S_i \cap S_j^*$.
Furthermore, since the $w^*_r$'s are standard Gaussians and the $w_r$'s are supposed to behave like standard Gaussians, we also expect the $v_r$'s to behave like standard Gaussians (except for the constancy on $S_i \cap S_j^*$ of course).
Defining $z_r^{*j} = 1$ if $y_r \in S_j^*$ and $0$ otherwise we introduce the new variables $z_r^{i,j} = z_r^i \cdot z_r^{*j}$.
Using composability of SoS proofs, we will show that $\cA$ implies the following system $\cB$ and hence, it will be enough to show that this one implies the desired error guarantee as well.
\footnote{Note that in general $\cA \sststile{t}{} \cB$ and $\cB \sststile{t'}{} p \geq 0$ only imply $\cA \sststile{t \cdot t'}{} p \geq 0$. However, a simple observation shows that if the second SoS proof only multiplies a constant number of constraints together and $t' = \cO(t)$ we can conclude $\cA \sststile{\cO(t)}{} p \geq 0$.}
\begin{align*}
\cB \coloneqq 	\left\{
					\begin{array}{l}
						\forall r \in [n], i,j \in [k]: z_r^{i,j} = (z_r^{i,j})^2 \\
						\forall r \in [n] \sum_{i=1}^{k} \sum_{j=1}^k z_r^{i,j} = 1 \\
						\forall i,j \in [k], r,s \in [n]: z_r^{i,j} z_s^{i,j} (v_r - v_s) = 0 \\
						\forall u \in \R^d: \Esymb_{r \sim [n]} [ \iprod{ v_r, u }^t ] \leq 2 \cdot (4t)^{t/2} \norm{u}_2^t
					\end{array}
				\right\}
\end{align*}

Note, that given any solution to the system $\cA$ we are not able to actually write down the values for these variables since they depend on the true values of the $w_r^*$s and $\mu_j^*$s.
However, we will be able to prove some properties about them which will help us in the proof of our main theorem.
It is is easily seen that the first three of the above constraints are implied by $\cA$ using the previous discussion.
To see why the last equation holds, we use that the $w_r$ and $w_r^*$ both satisfy the $t$-explicitly boundedness constraint and there is a SoS proof for this - see Section~\ref{sec:toolkit} for a more thorough definition.
We can then use the SoS Triangle Inequality to conclude.
Further, we there show as well that in the subguassianity constraint we can also instantiate $u$ as having entries of polynomial functions of our remaining variables.

It is hence enough to show that $\cB \sststile{\cO(t)}{} \frac{1}{\card{S}} \sum_{r \in S} \snormt{\mu_{k^*_r}^* - \mu_{k_r}} \leq 8 k^{(4+2c)/t} \cdot t$.
Our main tool will be the following technical lemma giving us a bound on the expected $t$-norm of our vectors $v_r$.
We defer the proof to the end of this section.
In the following let $S \subseteq \brac{n}$ such that $\card{S} \geq \frac{n}{k^c}$ for some $c \in \N_{\geq 0}$.
It might be instructive to think of $S = \brac{n}$ when reading the proof as this case already contains all of the main ideas.
The more general case can be easily obtained from here.
By $r \sim S$ we denote sampling an index $r \in S$ uniformly at random.
\begin{lemma}
\label{main_lemma}
For all $i,j \in [k]$ and even $t$ it holds that $\cB \sststile{\cO(t)}{} \{\Esymb_{r \sim S} z_r^{i,j} \norm{v_r}_2^t \leq 2 \cdot k^c \cdot (4t)^{t/2}\}$.
\end{lemma}

Writing our error in terms of the refinement $S_i \cap S_j^*$ and letting $\gamma = \Paren{ 2k^{2+c} \cdot (4t)^{t/2} }^{-\frac{2}{t}} > 0$ we get:
\begin{align*}
&\cB
\sststile{\cO(t)}{} \frac{1}{\card{S}} \sum_{r \in S} \norm{\mu_{k^*_r}^* - \mu_{k_r}}_2^2 
= \frac{1}{\card{S}} \sum_{i,j = 1}^k \sum_{r \in S} z_r^{i,j} \norm{\mu_i - \mu_j^*}_2^2 = \sum_{i,j = 1}^k \frac{1}{\card{S}} \sum_{r \in S} z_r^{i,j} \norm{v_r}_2^2 \\
&\leq \sum_{i,j=1}^k \frac{1}{\card{S}} \sum_{r \in S} \Paren{ \frac{2}{t} \gamma^{\frac{t}{2}-1} z_r^{i,j} \norm{v_r}_2^t + \frac{t-2}{t} \frac{1}{\gamma} z_r^{i,j} } = \frac{2}{t} \gamma^{\frac{t}{2}-1} \sum_{i,j=1}^k \Esymb_{r \sim S} z_r^{i,j} \norm{v_r}_2^t + \frac{t-2}{t} \frac{1}{\gamma} \frac{1}{\card{S}} \sum_{r \in S} \sum_{i,j=1}^k z_r^{i,j} \\
&\leq \frac{2}{t} \gamma^{\frac{t}{2}-1} ( 2 k^{c+2} \cdot (4t)^{t/2} ) + \frac{t-2}{t} \frac{1}{\gamma} = \frac{1}{\gamma} = 4 \sqrt[t]{4} \cdot k^{(4+2c)/t} \cdot t \leq 8k^{(4+2c)/t} \cdot t
\end{align*} 
where for the first inequality we used Lemma~\ref{toolkit:lemma_3} with $\gamma$ as above and for the second Lemma~\ref{main_lemma}.

We will finish this section by giving the proof of Lemma~\ref{main_lemma}:
\begin{proof}[Proof of Lemma~\ref{main_lemma}]
From the $t$-explicitly boundedness constraint in $\cB$ and since $t$ is even we know that
\begin{align*}
\forall u \in \R[z^{i,j},v_r]_{\leq 2}: \cB \sststile{\cO(t)}{} \Esymb_{r \sim S} \iprod{ v_r, u }^t =  \frac{n}{n} \cdot \frac{1}{\card{S}} \sum_{r \in S} \iprod{ v_r, u }^t \leq \frac{n}{\card{S}} \cdot \frac{1}{n} \sum_{r = 1}^n \iprod{ v_r, u }^t \leq k^c \cdot 2 (4t)^{t/2} \norm{u}_2^t
\end{align*}
where by $\R[z^{i,j},v_r]_{\leq 2}$ we denote polynomials in $z^{i,j}$ and $v_r$ of degree at most 2.
Fix $i,j \in [k]$ and fix $s \in [n]$.
Since $t$ is even, choosing $u = z_s^{i,j} v_s$ yields the following lower bound:
\begin{align*}
&\cB \sststile{\cO(t)}{} \frac{1}{\card{S}} \sum_{r \in S} \iprod{ v_r, u }^t \geq \frac{1}{\card{S}} \sum_{r \in S} z_r^{i,j }\iprod{ v_r, u }^t = \frac{1}{\card{S}} \sum_{r \in S}^n z_r^{i,j} z_s^{i,j} \iprod{ v_r, v_s }^t = \frac{1}{\card{S}} \sum_{r \in S} z_r^{i,j} z_s^{i,j} \norm{v_r}_2^t \norm{v_s}_2^t \\
&= z_s^{i,j} \norm{v_s}_2^t \cdot \Paren{ \Esymb_{r \sim S} z_r^{i,j} \norm{v_r}_2^t }
\end{align*} 
For the first inequality we used that $t$ is even and that $\cB \sststile{}{} z_r^{i,j} \leq 1$.
For the second equality we used the equality $z_r^{i,j} z_s^{i,j} (v_r - v_s) = 0$.
Combining this with the upper bound from our axioms and using that $ (z_s^{i,j})^2 = z_s^{i,j}$, we conclude:
\begin{align*}
\cB \sststile{\cO(t)}{} z_s^{i,j} \norm{v_s}_2^t \cdot \Paren{ \Esymb_{r \sim S} z_r^{i,j} \norm{v_r}_2^t }
\leq k^c \cdot 2(4t)^{t/2} \norm{z_s^{i,j} v_s}_2^t = k^c \cdot 2 (4t)^{t/2} z_s^{i,j} \norm{v_s}_2^t
\end{align*}

Averaging over all $s \in S$ we get:
\begin{align*}
\cB \sststile{\cO(t)}{} \Paren{ \Esymb_{r \sim S} z_r^{i,j} \norm{v_r}_2^t } \cdot  \Paren{ \Esymb_{r \sim S} z_r^{i,j} \norm{v_r}_2^t } \leq k^c \cdot 2 (4t)^{t/2} \cdot \Paren{ \Esymb_{r \sim S} z_r^{i,j} \norm{v_r}_2^t }
\end{align*}
and since $\cB \sststile{\cO(t)}{} \Esymb_{r \sim S]} z_r^{i,j} \norm{v_r}_2^t = \frac{1}{\card{S}} \sum_{r \in S} z_r^{i,j} \norm{v_r}_2^t  \geq 0$ we can apply Lemma~\ref{toolkit:lemma_2} with $C = k^c \cdot 2(4t)^{t/2}$ to conclude that $\cB \sststile{\cO(t)}{} \Esymb_{r \sim S} z_r^{i,j} \norm{v_r}_2^t  \leq k^c \cdot 2 (4t)^{t/2}$.
\end{proof}

To complete the proof of Theorem~\ref{thm:clustering_subset} what is left to show is that in the sum-of-squares proof we gave we have indeed $n_{\text{max}} = \mathrm{poly}(n)$.
As for the above proofs we can see that all terms are of degree $\cO(1)$ except for the $\normt{v_r}^t$ and the $\iprod{v_r, z_r^{i,j} v_s}$ for $r, s \in \brac{n}$ and $i, j \in \brac{k}$.
Hence, in total these are at most $\cO((kn)^2) d^{\cO(t)} = \mathrm{poly}(n)$.
Furthermore, all the proofs moved to the appendix also satisfy this property.

\subsection{Obtaining the final clustering}
\label{sub_sec:clustering_final}

Next, we will turn to proving Theorem~\ref{thm:clustering_final}.
Recall that we are now given a mean estimate $\hat{\mu}_r$ for each $r \in \brac{n}$ such that Theorem~\ref{thm:clustering_subset} holds.
Using these, we will first show a very simple procedure to obtain estimates for the clusters and then use our Robust Mean Estimation Algorithm to obtain our final mean estimates.
More precisely, we will first proof the following Lemma:
\begin{lemma}
\label{lem:intermediate_clusters}
With probability at least $1-\frac{1}{\mathrm{poly}(k)}$ we can obtain a set of clusters $\hat{S}_1, \ldots, \hat{S}_k$ such that for each $i \in [k]$ we have that $S_i$ contains at least a $1 - \frac{1}{\mathrm{poly}(k)}$ fraction of the points in $S_i^*$
\end{lemma}
To proof this, consider the following procedure that given a mean estimate $\hat{\mu}_r$ for each data point does the following (recall that $\Delta$ is the minimum separation of the real means, i.e., $\Delta = \min_{i \neq j \in [k]} \normt{\mu_i^* - \mu_j^*}$):
\begin{itemize}
\item Choose one of the remaining mean estimates uniformly at random.
\item Combine all indices whose mean estimates lie in a radius of $\Delta/4$ of the point previously chosen into a new cluster and remove them from the list of remaining points.
\item Repeat as long as there are at least $\frac{n}{2k}$ points remaining, else assign the remaining points arbitrarily and arbitrarily balance the clusters such that they have the same size.
\end{itemize}

This procedure is also inspired by the last part of the algorithm given by Hopkins and Li~\cite{Hopkins_Li_clustering}.
The proof of Lemma~\ref{lem:intermediate_clusters} now is as follows:
\begin{proof}[Proof of Theorem~\ref{thm:clustering_final}]
Suppose that $\Delta \geq \sqrt{C k^{(4+2c)/t} \cdot t}$ for $C = 64 \cdot 8$, $t$ the parameter we chose when running our algorithm, and $c \geq 3$.
Considering the above procedure, we do the following:
We call an index $r \in [n]$ \emph{good} if the corresponding mean estimate is close to the real mean. More precisely, if $\normt{\hat{\mu}_r - \mu_r^*} \leq \Delta/8$.
We call an index $r \in [n]$ \emph{bad} if it is not good.
Further, we observe that with high probability all of the real clusters have size between $(1-\frac{1}{k^{c-1}}) \frac{n}{k}$ and $(1+\frac{1}{k^{c-1}}) \frac{n}{k}$ this follows from a simple Chernoff bound.
In all what follows we can hence condition on this event.

We now observe the following: If the index chosen in step 1 is good and is in say the real cluster $i$, then all good indices which also belong to cluster $i$ will also be chosen in step 2.
Also, all good indices belonging to cluster $j \neq i$ will not be chosen, since the real means are at least $\Delta$ apart.
Hence, if we only choose good indices, the only indices that could potentially be misclassified are those which are bad or those that are moved around in the last step.
However, it is easy to see, that the number of good indices we can potentially loose in this last step is at most the number of bad points.

We now invoke Theorem~\ref{thm:clustering_subset} with
\begin{align*}
T = \Set{r \suchthat \text{$\Normt{\mu_{k_r} - \mu_{k_r}^*}$ is among the $\frac{n}{k^c}$ largest values of $\Normt{\mu_{k_r'} - \mu_{k_r'}^*}$ for $r' \in \brac{n}$}}
\end{align*}
breaking ties arbitrarily if necessary.
Note, that we cannot explicitly construct $T$ but the conclusion of Theorem~\ref{thm:clustering_subset} of course still holds and yields that:
\begin{align*}
\Esymb_{r \sim T} \snormt{\mu_{k_r} - \mu_{k_r}^*} \leq 8 k^{(4+2c)/t} \cdot t = \Paren{\Delta/8}^2
\end{align*}

Since the average of the largest $\frac{n}{k^c}$ squared distances of estimated means to real means is at most $\Paren{\Delta/8}^2$ we can conclude that all other (non-squared) distances are also at most $\Delta/8$ and hence all $r \not\in T$ are good indices and there are at most $\card{T} = \frac{n}{k^c}$ bad ones.
Hence, if we only choose good indices the clusters we form have size at least $\frac{n}{k}(1-\frac{1}{k^{c-1}})$ and hence with some calculation we can conclude that for $c \geq 3$ we form exactly $k$ clusters.
Further, for each of these clusters at least a $\frac{ 1- k^{-(c-1)} }{ 1+k^{-(c-1)} } = 1 - \frac{1}{\mathrm{poly}(k)}$ fraction is classified correctly.
By iteratively distinguishing the cases whether the first, second, etc. index we choose is good or bad it is not hard to see that the event $A \coloneqq \Set{\text{One of the indices we choose is bad.}}$ occurs with probability at most $\frac{1}{\mathrm{poly}(k)}$.
\end{proof}

Hence, each of our cluster-estimates is an $\e$-corrupted sample of the Standard Gaussian distribution (which is certifiably 1-subgaussian) for $\e = 1/\mathrm{poly}(k)$.
Hence, invoking Theorem~\ref{thm:moment_estimation} with $t=2$ for each of these estimates, we get mean estimates $\hat{\mu}_1, \ldots, \hat{\mu}_k$ such that $\normt{\mu^*_i - \hat{\mu}_i} \leq \frac{1}{\mathrm{poly}(k)}$ for all $i \in [k]$ which finishes the proof of Theorem~\ref{thm:clustering_final}.

\section{Moment estimation}
\label{sec:moment_estimation}

Next, we would like to apply our techniques to the problem of robust moment estimation and give a proof of Theorem~\ref{thm:moment_estimation} and Theorem~\ref{thm:higher_moment_estimation}.
Unfortunately, in this case we can't just directly apply Theorem~\ref{thm:improved_sos_estimation} but have to do some more work.
However, this will only be a minor obstacle not requiring many new ideas.
We will address it at the very end of this section.
For now, we will content ourselves to find high-degree sum-of-squares proofs.

We begin by introducing some notation:
Let $X = \set{x^*_1, \ldots, x^*_n}$ be the uncorrupted sample of $\cD$ and $D$ be the uniform distribution over $X$.
Again, we assume that $\cD$ is $t$-certifiably $C$-subgaussian in the sense of Definition~\ref{def:cert_subgaussian}.
By standard concentration of measure arguments it will be enough to show that we can approximate the low-degree moments of $D$ well, since they in turn are very close to those of the original distribution $\cD$.
Overloading notation, we let $\mu^*$ and $\Sigma^*$ be the first two moments of $D$, i.e., $\mu^* = \frac{1}{n} \sum_{i=1}^n x^*_i$ and $\Sigma^* = \frac{1}{n} \sum_{i=1}^n (x^*_i - \mu^*) (x^*_i - \mu^*)^\top$.
Further, let $Y= \set{y_1, \ldots, y_n}$ be an $\e$-corrupted sample of $X$.
Following the paradigm of \emph{efficient algorithms from identifiability proofs} we propose the following system of polynomials equations in scalar-valued variables $w_i$, vector-valued variables $x_i$.
It is basically the same as in~\cite{Kothari_Steinhardt_Steurer_robust_moment_estimation}.
We also remark that the SoS proofs given here are highly inspired by theirs although one key difference is that we do no longer use any form of H\"older's Inequality since this would yield to many high-degree variables and prevent us from gaining an improved by applying Theorem~\ref{thm:linearization}.
We instead rely on a sum-of-squares version of the AM-GM Inequality (cf. Lemma~\ref{lemma:sos_am_gm}).
Further, in the following let $t \in \N$ be even.
\begin{align*}
 \cA \coloneqq	\left\{
					\begin{array}{l}
						\sum_{i=1}^n w_i = (1-\e) \cdot n \\
						\forall i \in \brac{n}: w_i^2 = w_i \\
						\forall i \in \brac{n}: w_i \cdot (y_i - x_i) = 0 \\
						\forall t' \in \brac{t/2}: \forall u \in \R^d: \frac{1}{n} \sum_{i=1}^n \iprod{u, x_i - \mu}^{2t'} \leq (Ct')^{t'} \cdot \iprod{u, \Sigma u}^{t'}
					\end{array}
				\right\}
\end{align*}
Where we define the variables $\mu = \frac{1}{n} \sum_{i=1}^n x_i$ and $\Sigma = \frac{1}{n} \sum_{i=1}^n \paren{x_i - \mu} \paren{x_i - \mu}^\top$.
Hence, they can be thought of as our guess of the first two moments of $D$.
\footnote{As long as $n^{\cO(1)}$ is at least as large as the maximum eigenvalue of $\Sigma^*$, the diameter of our sample is at most $n^{\cO(1)}$ with high probability and hence we can make the above system explicitly bounded by shifting everything as necessary.}

One can interpret this system as follows: The $w_i$ variables select a $(1-\e)$ fraction of the corrupted samples and the last constraint ensures that the selected variables satisfy the subgaussianity constraint.
Again, we refer to Section~\ref{sec:toolkit} for an explanation how we can model this form of $\forall$ constraints.

For the case of mean and covariance estimation we will now prove the following two theorems:
\begin{theorem}
\label{thm:moment_estimation_intermediate}
For all $u \in \R^d$ and $t \in \N$ even we have that
\begin{align*}
\cA \sststile{t}{} \Set{  \iprod{u, \mu - \mu^*}^2 \leq \delta_1^2 \iprod{u, \Sigma^* u} }
\end{align*}
where $\delta_1 = \cO(C^{1/2} t^{1/2}) \e^{1-1/t}$.
Further, it also holds that for all $u \in \R^d$ and $t \in \N$ even we have that
\begin{align*}
\cA \sststile{t}{} \Set{\iprod{u, \Paren{ \Sigma - \Sigma^* } u}^2 \leq  \delta_2^2 \iprod{u, \Sigma^* u}^2 }
\end{align*}
where $\delta_2 = \cO(C t) \e^{1-2/t}$.
\end{theorem}

In order to prove this we first need a supplementary lemma.
This is a weakening of the second part Theorem~\ref{thm:moment_estimation_intermediate} which we will later use combined with a bootstrapping technique.
\begin{lemma}
\label{lemma:covariances}
For all $u \in \R^d$ we have:
\begin{align*}
\cA \sststile{4}{} \Set{  \iprod{u, \Sigma u}^2 \leq \cO(1) \iprod{u, \Sigma^* u}^2 }
\end{align*}
\end{lemma}
\begin{proof}
We will show the following
\begin{align*}
\cA \sststile{2}{} \Set{ \Paren{\iprod{u, \Sigma u} - \iprod{u, \Sigma^* u}}^2 \leq \cO(\e) \Paren{\iprod{u, \Sigma u}^2 + \iprod{u, \Sigma^* u}^2}}
\end{align*}
The result can then be obtained as follows:
\begin{align*}
&\iprod{u, \Sigma u}^2  + \iprod{u, \Sigma^* u}^2 \leq \cO(\e) \iprod{u, \Sigma u}^2 + \cO(\e) \iprod{u, \Sigma^* u}^2 + 2 \iprod{u, \Sigma u}  \iprod{u, \Sigma^* u} \\
&\leq \Paren{\frac{1}{2} + \cO(\e)} \iprod{u, \Sigma u}^2 + \Paren{2+\cO(\e)} \iprod{u, \Sigma^* u}^2
\end{align*}
where we used the sum-of-squares version of the AM-GM Inequality, cf. Lemma~\ref{lemma:sos_am_gm}, and that $\sststile{4}{u,x} \iprod{u , \Sigma u} \geq 0$ and $\sststile{2}{u,x} \iprod{u , \Sigma^* u} \geq 0$.
After rearranging this yields:
\begin{align*}
\iprod{u, \Sigma u}^2 \leq \frac{1+\cO(\e)}{\frac{1}{2}-\cO(\e)} \iprod{u, \Sigma^* u}^2 = \cO(1) \iprod{u, \Sigma^* u}^2 
\end{align*}

To see why the first part is true, we proceed as follows:
Define $r_1, \ldots, r_n \in \Set{0,1}$ such that $\sum_{i=1}^n r_i = \paren{1-\e}n$ and $r_i \paren{y_i - x_i^*} = 0$ for all $i \in \brac{n}$, i.e., $r_i = 1$ if $y_i = x_i^*$ and $0$ otherwise.
This we can do since $Y$ is an $\e$-corrupted sample of $X$.
Further, define $z_i = r_i \cdot w_i$, because $\set{w_i^2 = w_i} \sststile{2}{w_i} \paren{1-z_i} \leq \paren{1-w_i} + \paren{1-r_i}$ we have $\cA \sststile{2}{} \Set{\sum_{i=1}^n \paren{1-z_i} \leq 2\e n}$ and $\cA \sststile{2}{} (1-z_i)^2 = 1-z_i$.

Writing out the definition of $\Sigma$ and $\Sigma^*$ and using the sum-of-squares H\"older and Triangle Inequality (c.f. Fact~\ref{toolkit:lemma_12}) for the first and second inequality respectively we get:
\begin{align*}
&\cA \sststile{4}{}  \Paren{\iprod{u, \Sigma u} - \iprod{u, \Sigma^* u}}^2 = \Paren{\frac{1}{n} \sum_{i=1}^n \iprod{u, (x_i - \mu)}^2 - \frac{1}{n} \sum_{i=1}^n \iprod{u, (x_i^* - \mu^*)}^2}^2 \\ 
&= \Paren{\frac{1}{n} \sum_{i=1}^n \paren{1-z_i} \Brac{\iprod{u, (x_i - \mu)}^2 - \iprod{u, (x_i^* - \mu^*)}^2}}^2 \\
&\leq \Paren{\frac{1}{n} \sum_{i=1}^n \paren{1-z_i}^2} \Paren{\frac{1}{n} \sum_{i=1}^n \Brac{\iprod{u, (x_i - \mu)}^2 - \iprod{u, (x_i^* - \mu^*)}^2}^2 } \\
&\leq 4\e \Brac{\Paren{\frac{1}{n} \sum_{i=1}^n \iprod{u, x_i - \mu}^4} + \Paren{\frac{1}{n} \sum_{i=1}^n \iprod{u, x_i^* - \mu^*}^4}} \leq \cO(\e) \Paren{ \iprod{u, \Sigma u}^2 + \iprod{u, \Sigma^* u}^2} 
\end{align*}
where in the last step we used subgaussianity.
\end{proof}

Next, we will prove Theorem~\ref{thm:moment_estimation_intermediate}.
\begin{proof}[Proof of Theorem~\ref{thm:moment_estimation_intermediate}]
We will do the proof for degree $2t$ SoS proofs because it is slightly more readable.
To obtain the result for degree $t$ SoS proofs is then straightforwardly attained by substitution.
Note the slightly difference on $t$ in the exponent of $\e$ this is due to the fact that we use a degree $2t$ SoS proof and should not lead to confusion.

Similarly as before and using $\Set{z_i^2 = z_i} \sststile{2}{} \Set{\paren{1-z_i}^2 = \paren{1-z_i}}$ we get:
\begin{align*}
&\cA \sststile{2t}{} \iprod{u, \mu - \mu^*}^2 = \Paren{ \Esymb_{i \sim \brac{n}} (1-z_i) \cdot \Brac{(1-z_i) \iprod{u, x_i - x_i^*}} }^2 \leq 2\e \cdot \Esymb_{i \sim \brac{n}} (1-z_i) \iprod{u, x_i - x_i^*}^2  \\
&= 2\e \cdot \Esymb_{i \sim \brac{n}} (1-z_i) \iprod{u, (x_i - \mu) - (x_i^* - \mu^*) - (\mu^* - \mu)}^2 \\
&\leq \cO(\e) \cdot \Paren{ \Esymb_{i \sim \brac{n}} (1-z_i) \iprod{u, x_i - \mu}^2 + \Esymb_{i \sim \brac{n}} (1-z_i) \iprod{u, x_i^* - \mu^*}^2 + 2\e \iprod{u, \mu^* - \mu}^2 } \\
&\leq \cO(\e)   \frac{1}{t}  \Paren{ 4 (t-1) \gamma \cdot \e + \gamma^{-(t-1)} \Esymb_{i \sim \brac{n}} \iprod{u, x_i - \mu}^{2t} + \gamma^{-(t-1)} \Esymb_{i \sim \brac{n}} \iprod{u, x_i^* - \mu^*}^{2t} } \\
&+ \cO(\e^2) \iprod{u, \mu^* - \mu}^2 \\
&\leq \cO(\e)  \cdot \frac{1}{t} \cdot \Paren{ 4 (t-1) \gamma \cdot \e + \gamma^{-(t-1)} (Ct)^t \iprod{u, \Sigma u}^t  + \gamma^{-(t-1)} (Ct)^t \iprod{u, \Sigma^* u}^t  } \\
&+ \cO(\e^2) \iprod{u, \mu^* - \mu}^2 
\end{align*}
where in the second-last step we used Lemma~\ref{toolkit:lemma_4} applied to each term in the expectation independently and int the last step subgaussianity.

Applying first Lemma~\ref{lemma:covariances}, note that we can do this since $t$ is even, and picking $\gamma = \cO(Ct) \e^{-\frac{1}{t}} \iprod{u, \Sigma^* u}$ we get:
\begin{align*}
&\cA \sststile{2t}{} \iprod{u, \mu - \mu^*}^2 \leq \cO(\e) \cdot \frac{1}{t} \cdot \Paren{ 4 (t-1) \gamma \cdot \e  + \gamma^{-(t-1)} \cO((Ct)^t) \iprod{u, \Sigma^* u}^t  } + \cO(\e^2) \iprod{u, \mu^* - \mu}^2  \\ 
&\leq \delta_1^2 \iprod{u, \Sigma^* u} + \cO(\e^2) \iprod{u, \mu^* - \mu}^2 
\end{align*}
And since $\frac{1}{1-\cO(\e^2) } = \cO(1)$ when $\e$ is small, rearranging terms and absorbing constants in the $\cO$-notation of $\delta_1$, yields $\cA \sststile{2t}{} \Set{ \iprod{u, \mu - \mu^*}^2 \leq \delta_1^2 \iprod{u, \Sigma^* u} } $.

The proof of the second part is almost identical and hence we skip some steps which remain the same.
Again, we will do the proof for degree $2t$ SoS proofs.
\begin{align*}
&\cA \sststile{2t}{} \iprod{u, \Paren{ \Sigma - \Sigma^* } u}^2 = \Paren{ \Esymb_{i \sim \brac{n}} (1-z_i) \Brac{\iprod{u, x_i-\mu}^2 - \iprod{u, x_i^* - \mu^*}^2 } }^2 \\
&\leq \cO(\e) \cdot \Paren{ \Esymb_{i \sim \brac{n}} (1-z_i) \iprod{u, x_i - \mu}^4  + \Esymb_{i \sim \brac{n}} (1-z_i) \iprod{u, x^0_i - \mu^*}^4 } \\
&\leq \cO(\e)  \frac{2}{t}  \Paren{ 4(t/2 - 1)\gamma \cdot \e + \gamma^{-(t/2 - 1)} \Esymb_{i \sim \brac{n}} \iprod{u, x_i - \mu}^{2t} + \gamma^{-(t/2 - 1)} \Esymb_{i \sim \brac{n}} \iprod{u, x_i^* - \mu^*}^{2t} } \\
&\leq \cO(\e) \cdot \frac{2}{t} \cdot \Paren{  4(t/2 - 1)\gamma \cdot \e +  \gamma^{-(t/2 - 1)} \cO((C t)^t) \iprod{u, \Sigma^* u}^t  } \leq \delta_2^2 \iprod{u, \Sigma^* u}^2  
\end{align*}
where we picked $\gamma = \cO((Ct)^2) \e^{-\frac{2}{t}} \iprod{u, \Sigma^* u}^2$. 
\end{proof}

To complete the proof of Theorem~\ref{thm:moment_estimation} we use some basic facts about pseudo-distributions introduced in Chapter~\ref{sec:preliminaries}.
We remark that except for our linearization technique this proof is the same as in~\cite{Kothari_Steinhardt_Steurer_robust_moment_estimation}.
The novelty we introduce lies in the previous proofs and in applying Theorem~\ref{thm:linearization}.

\begin{proof}[Proof of Theorem~\ref{thm:moment_estimation}]
Let $t \in \N$ be even.
First, we note that in the above SoS-proofs we have $m = \mathrm{poly}(n)$ and $M = \mathrm{poly}(m)$.
Where we recall that $M$ is the number of terms which consist of multiplying two or more constraints together.
Also, similar to the Clustering Problem it is not hard to see, that $n_{max} = \mathrm{poly}(n)$.
Hence, Theorem~\ref{thm:linearization} guarantees us that there is a linearized version of the sum-of-squares proofs using $n', m' = \mathrm{poly}(t,n)$ variables and constraints respectively and having only degree 3.
Let $\cA'$ and $f$ be the objects given by Theorem~\ref{thm:linearization} and $\zeta$ be a degree-3 pseudo-expectation satisfying $\cA'$.
By Fact~\ref{fact:pE_efficient_optimization} this can be found in time $n^{\cO(1)}$.
Define $\hat{\mu} \in \R^d$ as $\hat{\mu} \coloneqq \pE_\zeta \mu$ and $\hat{\Sigma} \coloneqq \pE_\zeta \Sigma$.

For the case of estimating the mean, let $u \in \R^d$ and set $p(\mu) = \delta_1^2 \iprod{u, \Sigma^* u} - \iprod{u, \mu - \mu^*}^2$.
We then have by Cauchy-Schwarz for pseudo-expectations
\begin{align*}
&\delta_1^2 \iprod{u, \Sigma^* u} - \iprod{u, \hat{\mu} - \mu^*}^2 = \delta_1^2 \iprod{u, \Sigma^* u} - \Paren{\pE_\zeta \iprod{u, \mu - \mu^*}}^2 \geq \delta_1^2 \iprod{u, \Sigma^* u} - \pE_\zeta \iprod{u, \mu - \mu^*}^2 \\
&= \pE_\zeta p(\mu) = \pE_\zeta p'(f(\mu)) \geq 2^{-n^{\Omega(1)}}
\end{align*}
as we will see, the error term $2^{-n^{\Omega(1)}}$ can be absorbed into $\delta_1^2$ for our choices of $u$.
And we hence work with the guarantee $\iprod{u, \hat{\mu} - \mu^*}^2 \leq \delta_1^2 \iprod{u, \Sigma^* u}$.

Taking squares roots in this inequality and choosing $u = \frac{\hat{\mu} - \mu^*}{\normt{\hat{\mu} - \mu^*}}$ we conclude:
\begin{align*}
\normt{\hat{\mu} - \mu^*} \leq \delta_1 \sqrt{\iprod{u, \Sigma^* u}} \leq \delta_1 \max_{\normt{v} = 1} \sqrt{\iprod{v, \Sigma^* v}} = \delta_1 \max_{\normt{v} = 1} \normt{\Paren{\Sigma^*}^{1/2} v} = \delta_1 \norm{{\Sigma^*}^{1/2}} = \delta_1 \norm{\Sigma^*}^{1/2}
\end{align*}
where we used that $\Sigma^*$ as a covariance matrix is positive semi-definite.

Choosing $u = \frac{\Paren{\Sigma^*}^{-1}\Paren{\hat{\mu} - \mu^*}}{\normt{\Paren{\Sigma^*}^{-1/2}\Paren{\hat{\mu} - \mu^*}}}$ we get:
\begin{align*}
\normt{\Paren{\Sigma^*}^{-1/2}\Paren{\hat{\mu} - \mu^*}} \leq \delta_1 \frac{ \sqrt{ \iprod{\Paren{\Sigma^*}^{-1/2}\Paren{\hat{\mu} - \mu^*}, {\Paren{\Sigma^*}^{-1/2}\Paren{\hat{\mu} - \mu^*}}} } }{\normt{\Paren{\Sigma^*}^{-1/2}\Paren{\hat{\mu} - \mu^*}}} = \delta_1
\end{align*}
as desired.

To see why the last part is true, just note that for $\delta_2 =  \cO(C t) \e^{1-2/t}$ we have by Theorem~\ref{thm:moment_estimation_intermediate} and Cauchy-Schwarz for pseudo-expectations for all $u \in \R^d$:
\begin{align*}
\iprod{u, \Paren{\hat{\Sigma} - \Sigma^*} u}^2 = \Paren{ \pE_\zeta  \iprod{u, \Paren{\Sigma - \Sigma^*} u} }^2 \leq \pE_\zeta \iprod{u, \Paren{\Sigma - \Sigma^*} u}^2 \leq \delta_2^2 \iprod{u, \Sigma^* u}^2 + 2^{-n^{\cO(1)}}
\end{align*}
which, after taking square roots and using the assumption that $\lambda_{min}(\Sigma^*) \geq 2^{-n^{\cO(1)}}$, implies that
\begin{align*}
\iprod{u, (1 - 2\delta_2) \Sigma^* u} \leq \iprod{u, \hat{\Sigma} u} \leq \iprod{u, (1+2\delta_2) \Sigma^* u}
\end{align*}
and hence abusing the $\cO$-notation to hide the factor of $2$ we get $(1-\delta_2) \Sigma \sle \hat{\Sigma} \sle (1+\delta_2) \Sigma^*$ as desired.
\end{proof}

Since the proof for higher-order moments is almost identical, we move it to Section~\ref{sec:estimating_higher_moments}.
One key difference is that we force $u$ to be part of our variables since otherwise it is unclear how to obtain the sum-of-squares proof certifying the approximation guarantee of our moment estimates.
We further rely on Lemma~\ref{lemma:pEs_are_sos} to obtain this result.

\section{Information theoretic lower bounds}
\label{sec:lower_bounds}

For lower bounds regarding the outlier-robust moment estimation problem we refer the reader to~\cite{Kothari_Steinhardt_Steurer_robust_moment_estimation}.
For the Clustering Problem, we show that Theorem~\ref{thm:clustering_subset} is information theoretically optimal in the sense, that when optimizing over the choice of $t$ we can obtain the information theoretical lower bound of $c \log k$ for some constant $c$ up to the constant factor. We summarize this in the following Theorem.
\begin{theorem}
\label{lower_bound_thm}
There exists a uniform mixture of $k$ $d$-dimensional Gaussian distributions with unit covariance such that for any algorithm trying to estimate the means and the assignment of samples to means makes error at least $c \log k$ in expectation for some constant $c$.
Where the error is defined again as $\frac{1}{n} \sum_{r = 1}^n \norm{\mu_{k_r} - \mu^*_{k^*_r}}_2^2$, where $\mu^*_{k^*_r}$ is the true mean of the $r$-th sample and $\mu_{k_r}$ the one the algorithm outputs.
\end{theorem}
Note that we do not require the estimation algorithm to be efficient.

In order to prove Theorem~\ref{lower_bound_thm}, we will use Fano's Inequality which can be stated as follows and follows from~\cite{fano}:
\begin{theorem}[Fano's Inequality]
\label{Fano}
Given a sample $Z$ from a uniform mixture distribution with $m$ components let $J$ be the random variable that indicates which component $Z$ was generated from.
Consider any estimator (not necessarily efficiently computable) $\Psi$ that tries to predict $J$ given $Z$, we have the following lower bound:
\begin{align*}
\Psymb \Paren{ \Psi(Z) \neq J } \geq 1 - \frac{I(Z;J) + \log 2}{\log m}
\end{align*}
where $I(Z;J)$ is the mutual information between $Z$ and $J$.
The process of estimating $J$ given $Z$ is also called \emph{Hypothesis Testing}.
\end{theorem}

Consider a uniform mixture of the Gaussian Distributions $\cD_1, \ldots, \cD_k$ where $\cD_i = \cN(\lambda \cdot e_i, I_k)$.
Here $e_i$ denotes the $i$-th unit vector and $\lambda > 0$ will be chosen later.
To begin with, we can also assume that the algorithm is given the means of the clusters as this only decreases the failure probability.
Further, since the samples are drawn independently of each other and the algorithm knows the true means the task of predicting the assignment of the $i$-th sample is exactly the same as the above Hypothesis Testing Problem where the mixture distribution corresponds to the one we seek to estimate.

Following (\cite{wainwright_2019}, Chapter 15), we can then obtain the following Lemma:
\begin{lemma}
\label{lemma:lower_bound}
Let $p$ be the failure probability of the Hypothesis Testing Problem using the mixture distribution given above, then we can choose $\lambda = \Theta{\sqrt{\log k}}$ such that $p \geq c > 0$ for some absolute constant $c$.
\end{lemma}

As a corollary we immediately get a proof of Theorem~\ref{lower_bound_thm} by noticing that each time any algorithm wrongly predicts the cluster index for a sample it incurs an error of $\snormt{\lambda e_i - \lambda e_j} = 2\lambda^2$ and hence the total error in expectation is at least $c \log k$ for some constant c.

\begin{proof}[Proof of Lemma~\ref{lemma:lower_bound}]
Using the same notation as above we have that $p = \Psymb \Paren{ \Psi(Z) \neq J } \geq 1 - \frac{I(Z;J) + \log 2}{\log k}$.
We first deal with the case $k \geq 3$.
Representing the mutual information in terms of the Kullback-Leibler Divergence and using convexity we get that
\begin{align*}
I(Z;J) = \frac{1}{k} \sum_{j=1}^k \kldiv{\cD_j }{\cD} \leq \frac{1}{k^2} \sum_{j,l=1}^k \kldiv{\cD_j}{\cD_l}
\end{align*}
Applying standard results, we find that $\kldiv{\cD_j}{\cD_l} = \frac{1}{2} \norm{\lambda \cdot e_j - \lambda \cdot e_l}^2 = \lambda^2$ and hence $p \geq 1 - \frac{\lambda^2 + \log 2}{\log k}$.
Choosing $\lambda = \sqrt{\frac{\log k}{3}}$ we get $p \geq c$ for some constant $c$.

When dealing with the case $k = 2$, i.e., when we are in the regime of Binary Hypothesis Testing, the above inequality becomes meaningless since the right-hand side is negative.
To resolve this we revert to LeCam's Method (cf.~(\cite{wainwright_2019}, Chapter 15)) which states that $p \geq \frac{1}{2} \Paren{ 1 - TV(\cD_1,\cD_2) }$, where $TV(\cD_1,\cD_2)$ is the Total Variation Distance between $\cD_1$  and $\cD_2$.
By Pinsker's Inequality we have that $TV(\cD_1,\cD_2) \leq \frac{1}{2} \sqrt{ \kldiv{\cD_1}{\cD_2} } = \sqrt{\frac{\log 2}{12}} < 1$ for the same choice of $\lambda$.
Hence, also in this case, $p$ is lower bounded is bounded by some absolute constant and we can thus conclude the proof by taking the minimum of the two.
\end{proof}


\phantomsection
\addcontentsline{toc}{section}{References}
\bibliographystyle{amsalpha}
\bibliography{custom}

\newcommand{\etalchar}[1]{$^{#1}$}
\providecommand{\bysame}{\leavevmode\hbox to3em{\hrulefill}\thinspace}
\providecommand{\MR}{\relax\ifhmode\unskip\space\fi MR }
\providecommand{\MRhref}[2]{%
  \href{http://www.ams.org/mathscinet-getitem?mr=#1}{#2}
}
\providecommand{\href}[2]{#2}
\begin{thebibliography}{DKK{\etalchar{+}}19}

\bibitem[ARV09]{ARV_sparsest_cut}
Sanjeev Arora, Satish Rao, and Umesh Vazirani, \emph{Expander flows, geometric
  embeddings and graph partitioning}, Journal of the ACM (JACM) \textbf{56}
  (2009), no.~2, 1--37.

\bibitem[BKS15]{SoS_for_dictionary_learning_and_tensor_decomposition}
Boaz Barak, Jonathan~A. Kelner, and David Steurer, \emph{Dictionary learning
  and tensor decomposition via the sum-of-squares method}, Proceedings of the
  Forty-Seventh Annual ACM Symposium on Theory of Computing (New York, NY,
  USA), STOC 2015, Association for Computing Machinery, 2015, pp.~143--151.

\bibitem[BM16]{SoS_tensor_1}
Boaz Barak and Ankur Moitra, \emph{Noisy tensor completion via the
  sum-of-squares hierarchy}, Conference on Learning Theory, 2016, pp.~417--445.

\bibitem[BS]{SoS_online_tutorial}
Boaz Barak and David Steurer, \emph{Proofs, beliefs, and algorithms through the
  lens of sum-of-squares}.

\bibitem[BS14]{SoS_general_survey}
\bysame, \emph{Sum-of-squares proofs and the quest toward optimal algorithms},
  arXiv preprint arXiv:1404.5236 (2014).

\bibitem[CSV17]{charikar2017learning}
Moses Charikar, Jacob Steinhardt, and Gregory Valiant, \emph{Learning from
  untrusted data}, Proceedings of the 49th Annual ACM SIGACT Symposium on
  Theory of Computing, 2017, pp.~47--60.

\bibitem[DKK{\etalchar{+}}18]{diakonikolas_robust_gaussian}
Ilias Diakonikolas, Gautam Kamath, Daniel~M Kane, Jerry Li, Ankur Moitra, and
  Alistair Stewart, \emph{Robustly learning a gaussian: Getting optimal error,
  efficiently}, Proceedings of the Twenty-Ninth Annual ACM-SIAM Symposium on
  Discrete Algorithms, SIAM, 2018, pp.~2683--2702.

\bibitem[DKK{\etalchar{+}}19]{diakonikolas_robust}
Ilias Diakonikolas, Gautam Kamath, Daniel Kane, Jerry Li, Ankur Moitra, and
  Alistair Stewart, \emph{Robust estimators in high-dimensions without the
  computational intractability}, SIAM Journal on Computing \textbf{48} (2019),
  no.~2, 742--864.

\bibitem[DKS18]{Diakonikolas_robust_estimation_and_clustering}
Ilias Diakonikolas, Daniel~M. Kane, and Alistair Stewart, \emph{List-decodable
  robust mean estimation and learning mixtures of spherical gaussians},
  Proceedings of the 50th Annual ACM SIGACT Symposium on Theory of Computing
  (New York, NY, USA), STOC 2018, Association for Computing Machinery, 2018,
  pp.~1047--1060.

\bibitem[Fan49]{fano}
Robert~M Fano, \emph{The transmission of information}, Massachusetts Institute
  of Technology, Research Laboratory of Electronics, 1949.

\bibitem[GS12]{SoS_speed_up_FOCS_2012}
Venkatesan Guruswami and Ali Sinop, \emph{Faster sdp hierarchy solvers for
  local rounding algorithms}, Foundations of Computer Science, 1975., 16th
  Annual Symposium on (2012).

\bibitem[GW95]{GW_Max_Cut}
Michel~X Goemans and David~P Williamson, \emph{Improved approximation
  algorithms for maximum cut and satisfiability problems using semidefinite
  programming}, Journal of the ACM (JACM) \textbf{42} (1995), no.~6,
  1115--1145.

\bibitem[HL18]{Hopkins_Li_clustering}
Samuel~B. Hopkins and Jerry Li, \emph{Mixture models, robustness, and sum of
  squares proofs}, Proceedings of the 50th Annual ACM SIGACT Symposium on
  Theory of Computing (New York, NY, USA), STOC 2018, Association for Computing
  Machinery, 2018, pp.~1021--1034.

\bibitem[HSS15]{SoS_tensor_3}
Samuel~B Hopkins, Jonathan Shi, and David Steurer, \emph{Tensor principal
  component analysis via sum-of-square proofs}, Conference on Learning Theory,
  2015, pp.~956--1006.

\bibitem[HSS19]{fast_SoS_tensor_robust}
Samuel~B. Hopkins, Tselil Schramm, and Jonathan Shi, \emph{A robust spectral
  algorithm for overcomplete tensor decomposition}, Proceedings of Machine
  Learning Research, vol.~99, PMLR, 25--28 Jun 2019, pp.~1683--1722.

\bibitem[HSSS16]{fast_SoS_planted}
Samuel~B. Hopkins, Tselil Schramm, Jonathan Shi, and David Steurer, \emph{Fast
  spectral algorithms from sum-of-squares proofs: Tensor decomposition and
  planted sparse vectors}, p.~178–191, Association for Computing Machinery,
  New York, NY, USA, 2016.

\bibitem[KKK19]{list_decodable_regression}
Sushrut Karmalkar, Adam Klivans, and Pravesh Kothari, \emph{List-decodable
  linear regression}, Advances in Neural Information Processing Systems, 2019,
  pp.~7423--7432.

\bibitem[KKM18]{Kothari_robust_regression}
Adam~R. Klivans, Pravesh~K. Kothari, and Raghu Meka, \emph{Efficient algorithms
  for outlier-robust regression}, Conference On Learning Theory, {COLT} 2018,
  Stockholm, Sweden, 6-9 July 2018 (S{\'{e}}bastien Bubeck, Vianney Perchet,
  and Philippe Rigollet, eds.), Proceedings of Machine Learning Research,
  vol.~75, {PMLR}, 2018, pp.~1420--1430.

\bibitem[KSS18]{Kothari_Steinhardt_Steurer_robust_moment_estimation}
Pravesh~K. Kothari, Jacob Steinhardt, and David Steurer, \emph{Robust moment
  estimation and improved clustering via sum of squares}, Proceedings of the
  50th Annual ACM SIGACT Symposium on Theory of Computing (New York, NY, USA),
  STOC 2018, Association for Computing Machinery, 2018, pp.~1035--1046.

\bibitem[LRV16]{lai2016agnostic}
Kevin~A Lai, Anup~B Rao, and Santosh Vempala, \emph{Agnostic estimation of mean
  and covariance}, 2016 IEEE 57th Annual Symposium on Foundations of Computer
  Science (FOCS), IEEE, 2016, pp.~665--674.

\bibitem[MSS16]{SoS_tensor_2}
Tengyu Ma, Jonathan Shi, and David Steurer, \emph{Polynomial-time tensor
  decompositions with sum-of-squares}, 2016 IEEE 57th Annual Symposium on
  Foundations of Computer Science (FOCS), IEEE, 2016, pp.~438--446.

\bibitem[O'D17]{odonnell_bit_complexity}
Ryan O'Donnell, \emph{{SOS Is Not Obviously Automatizable, Even
  Approximately}}, 8th Innovations in Theoretical Computer Science Conference
  (ITCS 2017) (Dagstuhl, Germany) (Christos~H. Papadimitriou, ed.), Leibniz
  International Proceedings in Informatics (LIPIcs), vol.~67, Schloss
  Dagstuhl--Leibniz-Zentrum fuer Informatik, 2017, pp.~59:1--59:10.

\bibitem[RSS19]{sos_estimation_survey}
Prasad Raghavendra, Tselil Schramm, and David Steurer, \emph{High dimensional
  estimation via sum-of-squares proofs}, Proceedings of the International
  Congress of Mathematicians (ICM 2018) (Singapore) (Boyan Sirakov, Paulo
  Ney~de Souza, and Marcelo Viana, eds.), vol.~4, World Scientific Publishing
  Company, 2019, International Congress of Mathematicians 2018; Conference
  Location: Rio de Janeiro, Brazil; Conference Date: August 1-9, 2018, pp.~3389
  -- 3424 (en).

\bibitem[RV17]{Regev_lower_bound}
O.~{Regev} and A.~{Vijayaraghavan}, \emph{On learning mixtures of
  well-separated gaussians}, 2017 IEEE 58th Annual Symposium on Foundations of
  Computer Science (FOCS), 2017, pp.~85--96.

\bibitem[RW17]{raghavendra_bit_complexity}
Prasad Raghavendra and Benjamin Weitz, \emph{{On the Bit Complexity of
  Sum-of-Squares Proofs}}, 44th International Colloquium on Automata,
  Languages, and Programming (ICALP 2017) (Dagstuhl, Germany) (Ioannis
  Chatzigiannakis, Piotr Indyk, Fabian Kuhn, and Anca Muscholl, eds.), Leibniz
  International Proceedings in Informatics (LIPIcs), vol.~80, Schloss
  Dagstuhl--Leibniz-Zentrum fuer Informatik, 2017, pp.~80:1--80:13.

\bibitem[RY20]{list_decodable_1}
Prasad Raghavendra and Morris Yau, \emph{List decodable learning via sum of
  squares}, Proceedings of the 2020 {ACM-SIAM} Symposium on Discrete
  Algorithms, {SODA} 2020, Salt Lake City, UT, USA, January 5-8, 2020 (Shuchi
  Chawla, ed.), {SIAM}, 2020, pp.~161--180.

\bibitem[SS17]{fast_SoS_tensor}
Tselil Schramm and David Steurer, \emph{Fast and robust tensor decomposition
  with applications to dictionary learning}, Proceedings of Machine Learning
  Research, vol.~65, PMLR, 07--10 Jul 2017, pp.~1760--1793.

\bibitem[VW04]{single_linkage_clustering}
Santosh Vempala and Grant Wang, \emph{A spectral algorithm for learning mixture
  models}, Journal of Computer and System Sciences \textbf{68} (2004), no.~4,
  841 -- 860, Special Issue on FOCS 2002.

\bibitem[Wai19]{wainwright_2019}
Martin~J. Wainwright, \emph{High-dimensional statistics: A non-asymptotic
  viewpoint}, Cambridge Series in Statistical and Probabilistic Mathematics,
  Cambridge University Press, 2019.

\end{thebibliography}

\appendix


\section{Sum-of-squares toolkit}
\label{sec:toolkit}

In this section we will collect all the tools related to sum-of-squares proofs we used in this paper.
In particular, we include all the lemmas we used and a short description of how to model the $\forall$ constraints in Sections~\ref{sec:clustering} and~\ref{sec:moment_estimation}.
Although this is a by now standard technique, we include it nevertheless for completeness.

\subsection{Small lemmas}

Our main tool in this section will be a sum-of-squares version of the AM-GM Inequality:

\begin{lemma}[\cite{SoS_for_dictionary_learning_and_tensor_decomposition}, Lemma A.1]
\label{lemma:sos_am_gm}
\begin{align*}
\Set{w_1, \ldots, w_t} \sststile{t}{w} \Set{ \prod_{i=1}^t w_i \leq \frac{\sum_{i=1}^t w_i^t}{t} }
\end{align*}
\end{lemma}

\begin{lemma}
\label{toolkit:lemma_2}
For each real number $C > 0$ we have $\Set{X \geq 0, X^2 \leq C X} \sststile{2}{X} \Set{X \leq C}$.
\end{lemma}
\begin{proof}
Choosing $\gamma = 1 / \sqrt{C}$ and using the AM-GM inequality, we get:
\begin{align*}
X = (\gamma X) \cdot \frac{1}{\gamma} \leq \frac{1}{2} \cdot \gamma^2 X^2 + \frac{1}{2} \cdot \frac{1}{\gamma^2} \leq \frac{1}{2} \cdot \gamma^2 C X + \frac{1}{2} \cdot \frac{1}{\gamma^2} = \frac{1}{2} \cdot X + \frac{1}{2} \cdot C
\end{align*}
and rearranging yields the desired conclusion.
\end{proof}

\begin{lemma}
\label{toolkit:lemma_3}
For all real numbers $\gamma > 0$ and all natural numbers $t \geq 1$ it holds that $\Set{A \geq 0, B \geq 0} \sststile{t}{A,B} \Set{ A \cdot B \leq \frac{ \gamma^{t-1} A B^t + (t-1)\frac{1}{\gamma}A }{t} }$.
\end{lemma}
\begin{proof}
Since $A \geq 0$ is part of our axioms, it is enough to show that $\Set{B \geq 0} \sststile{t}{B} \Set{B \leq \frac{ \gamma^{t-1} B^t + (t-1)\frac{1}{\gamma} }{t} }$.
Writing $B = (\gamma^{\frac{t-1}{t}} B) \cdot (\gamma^{-\frac{1}{t}}) \cdot \ldots \cdot (\gamma^{-\frac{1}{t}})$ where we repeat the last factor $(t-1)$ times and again using the AM-GM Inequality, we conclude:
\begin{align*}
\cA
\vdash_t B
= (\gamma^{\frac{t-1}{t}} B) \cdot (\gamma^{-\frac{1}{t}}) \cdot \ldots \cdot (\gamma^{-\frac{1}{t}})
\leq \frac{ \gamma^{t-1}B^t + (t-1)\frac{1}{\gamma} }{t}
\end{align*}
\end{proof}

\begin{lemma}
\label{toolkit:lemma_4}
For all real numbers $\gamma > 0$ and all natural numbers $t \geq 1$ it holds that $\Set{A \geq 0, B \geq 0, A^2 = A } \sststile{t}{A,B} \Set{A \cdot B \leq \frac{ (t-1)\gamma \cdot A + \gamma^{-(t-1)} \cdot B^t }{t} }$
\end{lemma}
\begin{proof}
\begin{align*}
&A \cdot B = A^{t-1} \cdot B = (\gamma^{1/t} A) \cdot \ldots \cdot (\gamma^{1/t} A) \cdot (\gamma^{-\frac{t-1}{t}} B) \leq \frac{(t-1)\gamma \cdot A^t + \gamma^{-(t-1)} \cdot B^t}{t} \\
&= \frac{(t-1)\gamma \cdot A + \gamma^{-(t-1)} \cdot B^t}{t}
\end{align*}
where we repeat the $(\gamma^{1/t} A)$ factor $t-1$ times.
The first and the last step use $\Set{A^2 = A} \sststile{t}{A} \{A^t = A\}$ (follows from Lemma~\ref{toolkit:lemma_7}) and the second step uses the AM-GM Inequality.
\end{proof}

\begin{lemma}
\label{toolkit:lemma_7}
For every two polynomials $p(x)$ and $q(x)$ where $\deg \Paren{q(x)} = s$ such that $\cA \sststile{t}{x} \Set{p(x) = 0}$ it also holds that $\cA \sststile{t+2s}{x} \Set{p(x)q(x) = 0}$.
\end{lemma}
\begin{proof}
Since sum-of-squares equality proofs compose linearly it is enough to show the statement for $q(x) = x^\alpha$ where $\alpha$ is a multi-index of cardinality $s$.
Since $\cA \sststile{t}{x} \Set{p(x) = 0}$ we have in particular that $\cA \sststile{t}{x} \Set{p(x) \geq 0}$ and $\cA \sststile{t}{x} \Set{-p(x) \geq 0}$ and hence
\begin{align*}
0 \leq \paren{-p(x)}(x^\alpha - 1)^2 + x^{2\alpha}p(x) + p(x) = 2x^\alpha p(x)
\end{align*}
and analogously we can show that $0 \geq 2x^\alpha p(x)$.
On remarking that all terms in the respective sum-of-squares proofs have degree at most $t + 2s$ we get the desired conclusion.
\end{proof}

\begin{lemma}[\cite{Kothari_Steinhardt_Steurer_robust_moment_estimation}, Lemma A.2]
\label{toolkit:lemma_9}
For all even $t$ we have $\sststile{X,Y}{t} \Paren{X + Y}^t \leq 2^{t-1} \Paren{X^t + Y^t}$
\end{lemma}

The next lemma is very similar to Proposition A.5 in~\cite{Kothari_Steinhardt_Steurer_robust_moment_estimation} and its proof is heavily inspired by the proof given there.
However, we restate it here in a bit more generality needed to fit our needs.
We remark that for the approximate version we hide some of the relevant constant in the $\cO$ notation but that these can be chosen so that this works out.
\begin{lemma}
\label{lemma:pEs_are_sos}
Let $x \in \R^{d_1}, u \in \R^{d_2}$ for $d_1, d_2 \in \N$, $\cA = \Set{q_1(x) \geq 0, \ldots, q_m(x) \geq 0}$ and $\cB = \Set{r_1(u) = 0, \ldots, r_l(u) = 0}$ be a system of polynomial inequalities in $x$ and $u$ respectively.
Let $\mu$ be a degree-$2t$ pseudo-distribution (over $x$) that satisfies $\cA$ and suppose that $\cA \cup \cB \sststile{2t}{x,u} p(x,u) \geq 0$.
Then $\cB \sststile{2t}{u} \pE_\mu p(x,u) \geq 0$.
Furthermore, if $\mu$ only approximately satisfies $\cA$ then for $\e = 2^{-d_1^{\Theta(t)}}$ we have $\cB \sststile{2t}{u} \pE_\mu p(x,u) \geq \e \normt{u}^{2t}$.
If additionally $\cB \sststile{2t}{u} \normt{u}^{2t} \leq 2^{d_1^{\cO(t)}}$ we get $\cB \sststile{u}{2t} \pE_\mu p(x,u) \geq \e$.
\end{lemma}
\begin{proof}
We first consider the case when $\mu$ satisfies $\cA$ exactly.
Let $p_i(x,u) = \sum_{\alpha_1, \alpha_2} p_{\alpha_1,\alpha_2} x^{\alpha_1} u^{\alpha_2}$ be a polynomial of degree at most $s \leq t$.
Then we can write
\begin{align*}
p_i(x,u)^2 = \Paren{\sum_{\alpha_2} p_{\alpha_2}(x) u^{\alpha_2}}^2 = \iprod{l_i(x), u^{\otimes s}}^2
\end{align*}
where $p_{\alpha_2}(x) = \sum_{\alpha_1} p_{\alpha_1,\alpha_2} x^{\alpha_1}$ and $l_i(x) = \Paren{p_{\alpha_2}(x)}_{\alpha_2}$ is a vector with entries polynomial in $x$ of degree at most $s$.
Hence, if $b(u,x)$ is a sum-of-squares polynomial in $x$ and $u$ of degree at most $s \leq t$ we can write $b(u,x) = \iprod{L(x) u^{\otimes s}, L(x) u^{\otimes s}}$ for $L(x)$ a matrix with entries polynomial in $x$ of degree at most $s$.
It follows that $\pE_\mu b(u,x) = (u^{\otimes s})^\top \pE_\mu L(x)^\top L(x) u^{\otimes s}$.
Further, for every vector $v$ we have $v^\top \pE_\mu L(x)^\top L(x) v = \pE_\mu \Snormt{L(x)v} \geq 0$ since $\mu$ is of degree $2t$.
Hence, $\pE_\mu L(x)^\top L(x)$ is positive semi-definite and it follows that $\pE_\mu b(u,x)$ is a sum-of-squares in $u$ of degree at most $2s$.
Moreover, let $b(u,x)$ be a sum-of-squares polynomial such that $\deg(b q_j) \leq 2t$ for $j \in \brac{m}$.
Then by an analogous argument and since $\mu$ satisfies $\cA$ we get that  $\pE_\mu b(u,x) q_j(x) = (u^{\otimes t})^\top \pE_\mu q_j(x) L(x)^\top L(x) u^{\otimes t}$ and the matrix $\pE_\mu q_j(x) L(x)^\top L(x)$ is positive semi-definite.
Hence also $\pE_\mu b(u,x) q_j(x)$ is a sum-of-squares in $u$.

Since $\cA \cup \cB \sststile{2t}{u,x} p(x,u) \geq 0$ we can write $p(x,u) = b(x,u) + \sum_{j=1}^m b_j(x,u)q_j(x) + \sum_{k=1}^l b'_k(x,u)r(u)$ where $b, b_1, \ldots, b_m$ are sum-of-squares and $b'_1, \ldots, b'_l$ are normal polynomials.
Further, each term has degree at most $2t$.
For simplicity we only consider the case where $\card{S} \leq 1$ in the definition of sum-of-squares proofs, the more general case works analogously.
We then conclude that
\begin{align*}
\pE_\mu p(x,u) = \pE_\mu b(x,u) + \sum_{j=1}^m \pE_\mu b_j(x,u) q_j(x) + \sum_{k=1}^l \Paren{\pE_\mu b'_k(x,u)} r_k(u)
\end{align*}
which means exactly that $\cB \sststile{2t}{u} \pE_\mu p(x,u) \geq 0$ since $\pE_\mu b'_k(x,u)$ is also a polynomial in $u$.

Next, we consider the case when $\mu$ only approximately satisfies $\cA$.
Recall that this means that for all $S \sse \brac{m}$ and sum-of-squares polynomials $h$ such that $\deg (h \cdot \prod_{j \in S} q_j) \leq 2t$ we have $\pE_\mu h \cdot \prod_{j \in S} q_j \geq - \eta \normt{h} \prod_{j \in S} \normt{q_j}$ for $\eta = 2^{-d_1^{\Omega(t)}}$.
For the first part, nothing changes, thus, let $j \in \brac{m}$, then from above we know that we can write $\pE_\mu q_j(x) b_j(u,x) = (u^{\otimes s})^\top \pE_\mu q_j(x) L_j(x)^\top L_j(x) u^{\otimes s}$.
And for $\mu$ approximately satisfying $\cA$ we get that for $v$ a unit eigenvector of $\pE_\mu q_j(x) L(x)^\top L(x)$ with eigenvalue $\lambda$ that
\begin{align*}
\lambda = v^\top \Paren{ \pE_\mu q_j(x) L_j(x)^\top L_j(x) } v = \pE_\mu q_j(x) \iprod{L_j(x), v}^2 \geq - \eta' \normt{b_j} \normt{q_j}
\end{align*}
for $\eta' = 2^{-d_1^{\Omega(t)}}$.
Again, assuming that the bit-complexity of our constraints and sum-of-squares proof is at most polynomial in $d_1$ or $d_1^t$ we get that all eigenvalues are at least $-2^{-d_1^{\Omega(t)}}$.
Hence, for $\e = 2^{-d_1^{\cO(t)}}$ we have that $ \Paren{ \pE_\mu q_j(x) L_j(x)^\top L_j(x) } + \e I$ is p.s.d. and thus $\pE_\mu q_j b_j(u,x) + \e \normt{u}^{2t}$ is sum-of-squares in $u$.
Abusing notation a bit is follows directly that $\cB \sststile{2t}{u} \pE_\mu p(x,u) \geq \e \normt{u}^{2t}$.
The stronger conclusion if $\cB$ also implies that $\snormt{u}$ is bounded is immediate.
\end{proof}

\begin{lemma}
\label{toolkit:lemma_11}
Let $x, y \in \R^d$ be variables of our proof system, then we have $\sststile{2}{x,y} \iprod{x,y}^2 \leq \snormt{x} \snormt{y}$.
\end{lemma}
\begin{proof}
We have
\begin{align*}
&2 \snormt{x} \snormt{y} - 2 \iprod{x,y}^2 = 2 \sum_{i,j=1}^d x_i^2 y_j^2 - x_i y_i x_j y_j = \sum_{i,j=1}^d x_i^2 y_j^2 - 2 x_i y_i x_j y_j + x_j^2 y_i^2 \\
&= \sum_{i,j=1}^d (x_i y_i - x_j y_j)^2
\end{align*}
\end{proof}

\begin{lemma}[Adaptation of~\cite{Hopkins_Li_clustering}, Fact A.6]
\label{toolkit:lemma_12}
Let $w_i, x_i$ for $i \in \brac{n}$ be variables of our proof system and $t$ even.
Further, let $p_i$ be polynomials of degree at most $l$, then we have $\Set{w_1^2 = w_1, \ldots, w_n^2 = w_n} \sststile{tl}{w,x} \Paren{\sum_{i=1}^n w_i p_i(x_i)}^t \leq \Paren{\sum_{i=1}^n w_i}^{t-1} \sum_{i=1}^n p(x_i)^t$
\end{lemma}

\subsection{How to model universally quantified constraints}

In this section, we would like to show how we can model constraints of the form $\forall u \in \R^d: p_1(x, u) \leq p_2(x,u)$ where $x$ is part our variables and $p_1$ and $p_2$ are polynomials in $x$ and $u$.
Examples of this are the $t$-explicitly boundedness constraint we used in Section~\ref{sec:clustering} and the certifiable subgaussianity constraint in Section~\ref{sec:moment_estimation}.

We can do this by not only requiring that the inequality is indeed true for every $u$, but also that there is a sum-of-squares proof of this fact (in variables $u$).
I.e., that there exists $t \in \N$ such that $\sststile{t}{u} p_1(x,u) \leq p_2(x,u)$.
We can phrase this problem as finding a positive semi-definite matrix $M \in \R^{(d+1)^{t/2} \times (d+1)^{t/2}}$ such that
\begin{align*}
p_2(x,u) - p_1(x,u) = \iprod{ (1,u)^{\otimes t/2}, M (1,u)^{\otimes t/2}} = \Snormt{M^{1/2} (1,u)^{\otimes t/2}}
\end{align*}
where we fix the degree $t$ of the SoS proof and $(1,u)^{\otimes t/2}$ denotes the vector containing all $u^\alpha$ for all multi-indices $\alpha$ of size at most $t/2$.
We can also view the above equation as an equation between polynomials just in $u$ - considering $x$ as "fixed".
Hence, to enforce this equality it is enough to enforce equality between the coefficients of these polynomials and require that $M$ is positive semi-definite.
Although this last requirement neither fits our model of constraints this does not pose a problem since we obtain our pseudo-expectation by solving a positive semi-definite program and can hence add it to this.
Hence, we add at most $d^{\cO(t)}$ constraints and $(d+1)^t = d^{\cO(t)}$ variables.
This in particular means that $m = \mathrm{poly}(n)$ continues to hold for the number of constraints we use in the Clustering and Robust Moment Estimation Problems.

Another aspect to note is the following: Suppose we would like to reason about the case when $u$ is not just a certain parameter but a polynomial function of our variables $x$, i.e., we have $u(x)_i = s_i(x)$ for some polynomial $s_i$.
Denoting by $M \coloneqq \max_{1 \leq i \leq d} \deg s_i $ it is not hard to see that $\sststile{t \cdot M}{x} p_1(x,u(x)) \leq p_2(x,u(x))$.

\subsection{Deferred proofs}

In this section, we show that the first system $\cA$ in Section~\ref{sec:clustering} implies the second $\cB$.
The first three lines of $\cB$ are seen to be true following some elementary calculations.
Thus, we will content us with showing that the last line holds, i.e., that for all variables satisfying $\cA$ there exists a sum-of-squares proof (in variables $u$) of the fact that $\Esymb_{r \sim \brac{n}} \brac{ \iprod{ v_r, u }^t } \leq 2 \cdot (4t)^{t/2} \normt{u}^t$.
Recall that we have $v_r = w_r^* - w_r = w_r^* - (y_r  - \mu_r)$
The proof goes as follows:
\begin{align*}
&\cA \sststile{t}{u} \Esymb_{r \sim \brac{n}} \brac{ \iprod{ v_r, u }^t } = \frac{1}{n} \sum_{r = 1}^n \iprod{w_r^* - (y_r - \mu_r), u}^t \leq 2^{t-1} \frac{1}{n} \sum_{r = 1}^n \iprod{w_r^*, u}^t + 2^{t-1} \frac{1}{n} \sum_{r = 1}^n \iprod{y_r - \mu_r, u}^t \\
&= 2^{t-1} \Esymb_{r \sim \brac{n}} \iprod{w_r^*, u}^t + 2^{t-1} \frac{1}{n} \sum_{r = 1}^n \sum_{j=1}^k z_r^j \iprod{y_r - \mu_r, u}^t \\
&= 2^{t-1} \Esymb_{r \sim \brac{n}} \iprod{w_r^*, u}^t + 2^{t-1} \frac{1}{k} \sum_{j=1}^k \frac{k}{n} \sum_{r = 1}^n z_r^j \iprod{y_r - \mu_r, u}^t \\
&\leq 2^{t-1} 2t^{t/2} \normt{u}^t  + 2^{t-1} \frac{1}{k} \sum_{j=1}^k 2 t^{t/2} \normt{u}^t = 2 \cdot (4t)^{t/2} \normt{u}^t
\end{align*}
where in the first inequality we used a sum-of-squares version of the Triangle Inequality (cf. Lemma~\ref{toolkit:lemma_9}) and for the second one our subgaussianity constraints in $\cA$ and standard properties of gaussian random variables.
We also note, that there are at most $(kd)^{\cO(t)} = \mathrm{poly}(n)$ variables of degree $t$ and all other variables have constant degree.

\section{Proof of linearization theorem}
\label{sec:linearization_proof}

Here, we will give the proof of Theorem~\ref{thm:linearization}.
The constraints we will add are once, a linearized version of the original constraints and twice, a set of constraints ensuring that the new variables behave as expected and in particular that the standard computation rules of exponentiation apply.
On a high level, if our original variables are $x \in \R^n$ and we introduce say the three variables $y_\alpha, y_\beta, y_\gamma$ for multi-indices $\alpha, \beta, \gamma$ with the aim of representing $x^\alpha, x^\beta, x^\gamma$ respectively.
Since a priori there is no relation between these new variables, we will enforce constraints of the form $\set{y_\alpha y_\beta y_\gamma = y_{\alpha + \beta + \gamma}}$ to simulate the behavior expected from the corresponding $x$ variables.

\begin{proof}[Proof of Theorem~\ref{thm:linearization}]
Suppose that $\cA = \{q_1(x) \geq 0, \ldots, q_m(x) \geq 0\}$ and recall that $\cA \sststile{t}{x} \{p(x) \geq 0\}$ means that there exist sum-of-squares polynomials $b_S$ for each $S \sse \brac{m}$ such that $p(x) = \sum_{S \sse \brac{m}} b_S(x) \prod_{i \in S} q_i(x)$ and the degree of each term is at most $t$.
By our assumption there are at most $M$ summands such that $b_S \not\equiv 0$ and $\card{S} > 1$.
By introducing the constraint $q_S \coloneqq \Paren{\prod_{i \in S} q_i(x)} \geq 0$ for all these $S$ we get that we can write the sum-of-squares proof as
\begin{align*}
p(x) = b_0(x) + \sum_{j=1}^{m} b_j(x) q_j(x) + \sum_{\mathclap{\substack{S \sse \brac{m},\ \card{S} > 1 \\ b_S \not\equiv 0}}} b_S(x) q_S(x)
\end{align*}
and the degree of each term is again at most $t$.
For convenience we define $M' = m + M$, $q_0(x) \equiv 1$, and re-index the $q_S$ by the numbers $m+1, \ldots, m+M$.
Let $q_\alpha^{(j)}$ denote the coefficients of $q_j$ for $j = 0,\ldots, M'$ respectively.
Since the $b_j$s are sum-of-squares polynomials we can write $b_j(x) = \sum_{k=1}^{l_j} \Paren{u_{k,j}(x)}^2$ for polynomials $u_{k,j}$.
For notational convenience we set $l \coloneqq \max_{j = 0, \ldots, M'} l_j$ and define $u_{j,k} \equiv 0$ for $k > l_j$.
Also, let $u_\alpha^{(j,k)}$ denote the coefficients of $u_{j,k}$ for $j = 0, \ldots, M'$ and $k = 1,\ldots,l$ respectively.
In the following, we will explicitly write down all coefficients up to degree $t$ even if they are $0$ to make the notation cleaner.
We then observe:
\begin{align*}
&\sum_{\abs{\alpha} \leq t} p_\alpha x^\alpha = p(x) 
= \sum_{j=0}^{M'} 
\Paren{ \sum_{k=1}^l \Paren{ \; \sum_{\mathclap{\abs{\alpha} \leq t}} u_\alpha^{(j,k)} x^\alpha }^2 } 
\Paren{ \; \sum_{\mathclap{\abs{\alpha} \leq t}} q_\alpha^{(j)} x^\alpha } \\
&= \sum_{j=0}^{M'} 
\Paren{ \;\;\; \sum_{\mathclap{\abs{\alpha}, \abs{\beta} \leq t}} x^{\alpha + \beta} \Paren{ \sum_{k=1}^l u_\alpha^{(j,k)} u_\beta^{(j,k)} } } 
\Paren{ \; \sum_{\mathclap{\abs{\alpha} \leq t}} q_\alpha^{(j)} x^\alpha } \\
&= \sum_{j=0}^{M'} \sum_{\vert \delta \vert \leq t} x^\delta \left( \;\;\;\; \sum_{\mathclap{\alpha + \beta + \gamma = \delta}} q_\gamma^{(j)} \Paren{ \sum_{k=1}^l u_\alpha^{(j,k)} u_\beta^{(j,k)} } \right) =  \; \sum_{\mathclap{\abs{\delta} \leq t}}x^\delta \Paren{  \sum_{\alpha + \beta + \gamma = \delta} \sum_{j=0}^{M'} \sum_{k=1}^l u_\alpha^{(j,k)} u_\beta^{(j,k)} q_\gamma^{(j)} }
\end{align*}
and hence by comparing coefficients we get $p_\delta = \sum_{\alpha + \beta + \gamma = \delta} \sum_{j=0}^{M'} \sum_{k=1}^l u_\alpha^{(j,k)} u_\beta^{(j,k)} q_\gamma^{(j)}$.
Note that we also used that there exist no terms of degree larger than $t$.

Next, consider the following multi-set of multi-indices $A = A_1 \cup A_2$ where
\begin{align*}
&A_1 = \Set{\alpha \sse \N^n \suchthat 1 \leq \card{\alpha} \leq t, \text{$\alpha$ has a non-zero coefficient in the above proof}} \\
&A_2 = \Set{\alpha + \beta + \gamma \suchthat \alpha,\beta,\gamma \in A_1}
\end{align*}
where for $A_1$ the non-zero coefficient may occur in either $p$ or one of the $u_{j,k}$ or $q_j$.
We stress that viewing it as a multi-set makes the reasoning easier and we will argue later why the gain in size is negligible for us.
We now define $f: \R^n \to \R^{n'}$ as $x \mapsto (x, (x^\alpha)_{\alpha \in A})$ where $n'  = n + \card{A} \leq n + \card{A_1} + \card{A_1}^3$.
By definition we have $f(x)_i = x_i$ for all $i \in \brac{n}$  and thus $f$ satisfies Point~\ref{linearization_1}.
Also note, that $f$ is clearly injective.
In the following, we will index coordinates $n+1$ to $n'$ of vectors in $\R^{n'}$ by multi-indices in $A$.

Next, we define a "linearized version" of the polynomials occurring in our original sum-of-squares proof, each taking as input a vector in $\R^{n'}$.
\begin{align*}
&p'(y) = \sum_{\mathclap{\abs{\alpha} \leq t}} p_\alpha y_\alpha&
&u_{k,j}'(y) = \sum_{\mathclap{\abs{\alpha} \leq t}} u_\alpha^{(j,k)} y_\alpha&
&b_j'(y) = \sum_{k=1}^l \Paren{u_{k,j}'(y) }^2& 
&q_j'(y) = \sum_{\mathclap{\abs{\alpha} \leq t}} q_\alpha^{(j)} y_\alpha&
\end{align*}
for $j = 0, \ldots, M'$ and $k = 1,\ldots,l$.
In case one multi-index might occur more than once (because we defined $A$ as a multi-set and there might be some overlap with the first $n$ coordinates), we use only one of them chosen arbitrarily.
Clearly, we have for all $x \in \R^n$ that $p'(f(x)) = \sum_{\abs{\alpha} \leq t} p_\alpha x^\alpha = p(x)$, thus satisfying Point~\ref{linearization_3}.
We note that this analogously holds for the other polynomials above as well.
As the reader might have already guessed the idea behind this is that these "linearized polynomials" will constitute the sum-of-squares proof that $p'$ is non-negative.

Next, we will construct the constraints and the sum-of-squares proof.
First, we introduce the following set of constraints $\cA'_1 = \Set{q_1'(y) \geq 0, \ldots, q'_{M'}(y) \geq 0}$.
As we will see later this is not yet quite enough.
However, we will first start by trying to find a sum-of-squares proof that $p'(y) \geq 0$ and it will than be obvious which further constraints we need to add.
Hence, guessing which form this will take, we compute:
\begin{align*}
&\sum_{j=0}^{M'} b'_j(y) q'_j(y) = 
\sum_{j=0}^{M'} 
\Paren{ \;\;\; \sum_{\mathclap{\abs{\alpha}, \abs{\beta} \leq t}} y_\alpha y_\beta \Paren{ \sum_{k=1}^l u_\alpha^{(j,k)} u_\beta^{(j,k)} } } 
\Paren{ \sum_{\abs{\alpha} \leq t} q_\alpha^{(j)} y_\alpha }  \\
&= \sum_{\mathclap{\vert \delta \vert \leq t}} \;\;\;\; \sum_{\mathclap{\alpha + \beta + \gamma = \delta}} y_\alpha y_\beta y_\gamma \sum_{j=0}^{M'} \sum_{k=1}^l u_\alpha^{(j,k)} u_\beta^{(j,k)} q_\gamma^{(j)} \\
&= \sum_{\mathclap{\vert \delta \vert \leq t}} \;\;\;\; \sum_{\mathclap{\alpha + \beta + \gamma = \delta}} \Paren{y_{\alpha + \beta + \gamma} - \Paren{y_{\alpha +\beta + \gamma} - y_\alpha y_\beta y_\gamma  }}\sum_{j=0}^{M'} \sum_{k=1}^l u_\alpha^{(j,k)} u_\beta^{(j,k)} q_\gamma^{(j)} = \sum_{\vert \delta \vert \leq t} p_\delta y_\delta - S = p'(y) - S
\end{align*}
where we denote by $S$ all remaining terms.

Note that all of the terms in $S$ are of the form $u_\alpha^{(j,k)} u_\beta^{(j,k)} q_\gamma^{(j)} \paren{ y_{\alpha + \beta + \gamma} - y_\alpha y_\beta y_\gamma}$.
Hence, we introduce the constraints $\cA'_2 = \Set{y_{\alpha + \beta + \gamma} = y_\alpha y_\beta y_\gamma \suchthat \alpha, \beta, \gamma \in A_1}$ and define $\cA' = \cA'_1 \cup \cA'_2$.
Since we can write $p'(y) = \sum_{j=0}^{M'} b'_j(y) q'_j(y)  + S$ and all of the terms on the right-hand side have degree at most three we conclude that $\cA' \sststile{3}{y} \Set{p'(y) \geq 0}$.
Clearly, we have $m' = \card{\cA'} \leq M' + \card{A_1}^3 = m + M + \card{A_1}^3$.

What is left to do is to verify Point~\ref{linearization_2} and the bound on $n'$ and $m'$.
We will start with the former.
For $x \in \cA$ we clearly have that $f(x)$ satisfies the constraints in $\cA'_1$ since $q'_j(f(x)) = q_j(x) \geq 0$ for all $j \in \brac{m}$.
For the constraints in $\cA'_2$ just note that for $\alpha, \beta, \gamma \in A_1$ we have $\alpha + \beta + \gamma \in A_2 \sse A$ and hence $f(x)_{\alpha + \beta + \gamma} = x^{\alpha + \beta + \gamma} = x^\alpha x^\beta x^\gamma = f(x)_\alpha f(x)_\beta f(x)_\gamma$ as desired.
On the other hand, let $x \in \R^n$ such that $f(x) \in \cA'$.
By exactly the same reasoning as before we have that $q_j(x) = q'_j(f(x)) \geq 0$ for all $j \in \brac{m}$ and hence $x \in \cA$.

The bound on $n'$ and $m'$ is simple: Observe that $\card{A_1} = \sum_{i=1}^t n_i \leq t \cdot n_{\mathrm{max}}$ and $n \leq \sum_{i=1}^t n_i \leq t \cdot n_{\mathrm{max}}$.
And hence using the observations made before we conclude that
\begin{align*}
&n' \leq n + \card{A_1} + \card{A_1}^3 = n + \cO((t \cdot n_{\mathrm{max}})^3) = \cO((t \cdot n_{\mathrm{max}})^3) \\
&m' \leq m + M + \card{A_1}^3 = m + M + \cO((t \cdot n_{\mathrm{max}})^3)
\end{align*}
\end{proof}

\section{Estimating higher order moments}
\label{sec:estimating_higher_moments}

In this section, we will turn to the higher-order moments and the proof of Theorem~\ref{thm:higher_moment_estimation}.
We will first prove the following SoS version of the theorem:
\begin{theorem}
\label{thm:moment_estimation_higher_moments_intermediate}
For $t \in \N$ even and $r \in \N$ such that $r \leq t/2$ we have that
\begin{align*}
\cA \cup \Set{\iprod{u, M_2^* u} = 1} \sststile{t}{x,u} \Set{ \iprod{M_r - M_r^*, u^{\otimes r}}^2 \leq \delta_r^2 }
\end{align*}
where $\delta_r = \cO(C^{r/2} t^{r/2})  \e^{1- r/t}$.
\end{theorem}

The proof of this is very similar to the proof of Theorem~\ref{thm:moment_estimation_intermediate} and does not use any new ideas.
One key difference however is that this time we require $u$ to be part of our variables.
This stems from the fact, that we would also like to have a certificate - in form of a SoS proof - of our approximation guarantee, since the applications rely on it.
Again, we refer the interested reader to the original paper for more details.
\begin{proof}
Again, we will do the proof for degree $2t$ SoS proofs because it is slightly more readable.
Using the same techniques as in the previous section and skipping of some analogous steps we get:
\begin{align*}
&\cA \sststile{2t}{x,u} \Paren{\iprod{M_r - M_r^*, u^{\otimes r}}}^2 = \Paren{ \Esymb_{i \sim \brac{n}} (1-z_i) \Brac{\iprod{u, x_i-\mu}^r - \iprod{u, x_i^* - \mu^*}^r } }^2 \\
&\leq \cO(\e) \cdot \Paren{ \Esymb_{i \sim \brac{n}} (1-z_i) \iprod{u, x_i - \mu}^{2r}  + \Esymb_{i \sim \brac{n}} (1-z_i) \iprod{u, x_i^* - \mu^*}^{2r} } \\
&\leq \cO(\e)  \frac{r}{t}  \Paren{ 4(t/r - 1)\gamma \cdot \e + \gamma^{-(t/r - 1)} \Esymb_{i \sim \brac{n}} \iprod{u, x_i - \mu}^{2t} + \gamma^{-(t/r - 1)} \Esymb_{i \sim \brac{n}} \iprod{u, x_i^* - \mu^*}^{2t} } \\
&\leq \cO(\e) \cdot \frac{r}{t} \cdot \Paren{  4(t/r - 1)\gamma \cdot \e +  \gamma^{-(t/r - 1)} \cO((C t)^t) \iprod{u, \Sigma^* u}^t  } \leq \delta_r^2
\end{align*}
where we picked $\gamma =  \cO((Ct)^r) \e^{-\frac{r}{t}}$ and used that $\Set{\iprod{u, M_2^* u} = 1} \sststile{2}{u} \iprod{u, \Sigma^* u} \leq \iprod{u, M_2^* u} = 1$.
\end{proof}

The proof of Theorem~\ref{thm:higher_moment_estimation} now goes as follows.
\begin{proof}[Proof of Theorem~\ref{thm:higher_moment_estimation}]
Note, that this time we have two sets of variables, $x$ and $u$, and we will linearize only in $x$ - e.g., $x^\alpha u^\beta$ will become $x_\alpha u^\beta$.
To this end we define $\cB = \Set{r_1 = 0}$, where $r_1 = \iprod{u, M_2^* u} - 1$ and let $p(x,u) = \delta_r^2 - \iprod{M_r - M_r^*, u^{\otimes r}}^2$ and $p'(x,u)$ be its linearized version, again in $x$.
Let $\cA'$, and $f$ be the objects given by Theorem~\ref{thm:linearization} and let $\zeta$ be a degree-3 pseudo-distribution approximately satisfying $\cA'$, this can be found in time $n^{\cO(1)}$.
We then have that $\cA' \cup \cB \sststile{}{} p'(f(x),u) \geq 0$ where the degree in $f(x)$ is at most 3.
We define $M_r' = (M_r')_{\card{\alpha} = r}$ as $(M_r')_\alpha = \frac{1}{n} \sum_{i=1}^n (x_i)_\alpha$ and set $\hat{M}_r = \pE_\zeta M_r'$, where $(x_i)_\alpha$ are our linearized variables.

First, we note the following:
\begin{align*}
&\sststile{t}{u} \delta_r^2 - \iprod{\hat{M_r} - M_r^*, u^{\otimes r}}^2 \geq \pE_\zeta \delta_r^2 - \iprod{M_r' - M_r^*, u^{\otimes r}}^2 = \pE_\zeta p'(f(x),u)
\end{align*}
Where we used, that Cauchy-Schwarz for pseudo-expectations continues to hold in this setting where we have polynomials depending on $u$ and that this has a sum-of-squares proof in $u$.

Using the same notation as in the proof of Theorem~\ref{thm:linearization} by $\cA' \cup \cB \sststile{}{} p'(f(x),u) \geq 0$ we also get:
\begin{align*}
&\pE_\zeta p'(f(x),u) = \sum_{j=0}^m \pE_\zeta b_j'(f(x),u) q_j'(f(x)) + \pE_\zeta c(f(x),u) r_1(u) \\
&= \Paren{u^{\otimes t/2}}^\top L^\top L u^{\otimes t/2} + c(u) r_1(u) - \eta \normt{u}^{t/2}
\end{align*}
where $L$ is some matrix independent of $u$, $c(u)$ is some polynomial in $u$ and $\eta = 2^{-n^{\Omega(1)}}$.
Although the last step might seem very cryptic at first, its basically redoing the proof of Fact~\ref{lemma:pEs_are_sos} and hence we invite the reader to study this instead of reproducing it here.

Further, by our assumption that $n \geq t/2 \log \Paren{ \frac{1}{\lambda_{min}} }$ for $\lambda_{min} = \lambda_{min} (M_2^*)$ we have that
\begin{align*}
\Set{\iprod{u, M_2^* u} = 1} \sststile{t}{u} \normt{u}^{t/2} \leq \Iprod{u, \frac{1}{\lambda_{min}} M_2^* u}^{t/2} = \Paren{ \frac{1}{\lambda_{min}} }^{t/2} \leq 2^n
\end{align*}
and hence by combining all of the above statements above we get $\cB \sststile{t}{u} \delta_r^2 - \iprod{\hat{M_r} - M_r^*, u^{\otimes r}}^2 \geq 2^{-n^{\Omega(1)}}$.
Absorbing the error term in the $\cO$-notation of $\delta_r$ we have proven that $\cB \sststile{t}{u} \iprod{\hat{M}_r - M_r^*, u^{\otimes r}}^2 \leq \delta_r^2$.
\end{proof}

\end{document}